\documentclass[11pt,letterpaper]{article}
\usepackage[utf8]{inputenc}
\usepackage[T1]{fontenc}
\usepackage{lmodern}
\usepackage[DIV=11]{typearea}

\usepackage{microtype}
\usepackage{times}

\usepackage[utf8]{inputenc}
\usepackage{amssymb}
\usepackage{amsmath}
\usepackage{amsthm}
\usepackage[ruled,vlined,linesnumbered,nokwfunc]{algorithm2e_}
\usepackage{tikz}
\usetikzlibrary{arrows,shapes}
\usetikzlibrary{decorations,shadows}
\definecolor{lightgray}{rgb}{0.85,0.85,0.85}
\tikzstyle{player1}=[draw, thick, circle, fill=lightgray,inner sep=2pt, minimum width=12pt]
\tikzstyle{player2}=[draw, thick, fill=lightgray,inner sep=3pt, minimum width=15pt,minimum height=15pt]
\tikzstyle{p1}=[player1,minimum size=17pt]
\tikzstyle{p2}=[player2,minimum size=15pt]
\tikzstyle{arrow}=[->,line width=1pt,>=stealth',bend angle=15]

\usepackage{xcolor}
\usepackage{xspace}
\usepackage{booktabs}
\usepackage{nicefrac}
\usepackage{hyperref}
\usepackage{authblk}

\title{Improved Set-based Symbolic Algorithms for Parity Games}
\author[1]{Krishnendu Chatterjee}
\author[2,3]{Wolfgang Dvo{\v r}{\' a}k}
\author[3]{\\Monika Henzinger}
\author[4,3]{Veronika Loitzenbauer}

\affil[1]{IST Austria}
\affil[2]{TU Wien, Vienna, Austria}
\affil[3]{University of Vienna, Vienna, Austria} 
 \affil[4]{Bar-Ilan University, Ramat Gan, Israel}
 
\newcommand{\set}[1]{\{#1\}}
\newcommand{\lu}{\textup{(}}
\newcommand{\ru}{\textup{)}}
\newcommand{\upbr}[1]{\lu #1\ru}
\newcommand{\sseq}{\langle v_0,v_1,v_2,\ldots\rangle}
\newcommand{\pat}{\omega\xspace} 
\DeclareMathOperator{\Out}{\textit{Out}\xspace}
\newcommand{\at}[3]{\textit{Attr}_{#1}(#2, #3)\xspace}
\newcommand{\idxa}{i\xspace}
\newcommand{\straa}{\sigma\xspace}
\newcommand{\strab}{\pi\xspace}
\newcommand{\pe}{{\mathcal{E}}\xspace}
\newcommand{\po}{{\mathcal{O}}\xspace}
\newcommand{\prio}{\alpha\xspace}
\newcommand{\numprio}{c\xspace}
\newcommand{\ve}{V_\pe\xspace}
\newcommand{\vo}{V_\po\xspace}
\newcommand{\we}{W_\pe\xspace}
\newcommand{\wo}{W_\po\xspace}
\newcommand{\pl}{{z}\xspace}
\newcommand{\op}{{\overline{z}}\xspace}
\newcommand{\pgame}{\mathcal{P}\xspace}
\newcommand{\game}{\mathcal{G}\xspace}
\newcommand{\domsize}{h\xspace}
\newcommand{\dombound}[2]{h\xspace}
\DeclareMathOperator{\domalg}{\ProcNameSty{Dominion}\xspace}
\newcommand{\idx}{k}
\newcommand{\rank}{r}
\newcommand{\pre}{\mathsf{Pre}\xspace}

\newcommand{\cpre}[1]{\mathsf{CPre}_{#1}\xspace}
\newcommand{\numcpre}{\#\mathsf{CPre}}
\newcommand{\incr}{\mathrm{inc}}

\newcommand{\decr}{\mathrm{dec}}
\newcommand{\proj}[1]{\left< #1\right>}
\newcommand{\best}{\mathrm{best}}

\newcommand{\lift}{\operatorname{Lift}}

\theoremstyle{plain}
\newtheorem{theorem}{Theorem}

\newtheorem{lemma}[theorem]{Lemma}
\newtheorem{keylemma}[theorem]{Key Lemma}

\newtheorem{invariant}[theorem]{Invariant}
\newtheorem{claim}[theorem]{Claim}
\newtheorem{remark}[theorem]{Remark}
\newtheorem{example}[theorem]{Example}
\newcommand{\qee}{\phantom{a}\hfill$\lozenge$}

\begin{document}

\maketitle

\begin{abstract}
Graph games with $\omega$-regular winning conditions provide a
mathematical framework to analyze a wide range of problems in the
analysis of reactive systems and programs (such as the synthesis of reactive 
systems, program repair, and the verification of branching time properties). 
Parity conditions are canonical forms to specify $\omega$-regular winning
conditions. 
Graph games with parity conditions are equivalent to $\mu$-calculus model 
checking, and thus a very important algorithmic problem. 
Symbolic algorithms are of great significance because they provide scalable
algorithms for the analysis of large finite-state systems, as well as algorithms 
for the analysis of infinite-state systems with finite quotient.
A set-based symbolic algorithm uses the basic set operations and the 
one-step predecessor operators.
We consider graph games with $n$ vertices and parity conditions with $\numprio$
priorities (equivalently, a $\mu$-calculus formula with $\numprio$ alternations of 
least and greatest fixed points).
While many explicit algorithms exist for graph games with parity conditions,
for set-based symbolic algorithms there are only two algorithms
(notice that we use space to refer to the number of sets stored by a symbolic algorithm):
(a)~the basic algorithm that requires $O(n^\numprio)$ symbolic operations and 
linear space; and (b)~an improved algorithm that requires  $O(n^{\numprio/2+1})$ 
symbolic operations but also  $O(n^{\numprio/2+1})$ space (i.e., exponential space).
In this work we present two set-based symbolic algorithms for parity games:
(a)~our first algorithm requires  $O(n^{\numprio/2+1})$ symbolic operations and only
requires linear space; and (b)~developing on our first algorithm, we present 
an algorithm that requires  $O(n^{\numprio/3+1})$ symbolic operations and only
linear space.
We also present the first linear space set-based symbolic algorithm for parity games
that requires at most a sub-exponential number of symbolic operations.
\end{abstract}

\pagebreak
\section{Introduction}\label{sec:intro}
In this work we present improved set-based symbolic algorithms for 
solving graph games with parity winning conditions, which is equivalent to 
modal $\mu$-calculus model-checking.

\paragraph{Graph games.} 
Two-player graph games provide the mathematical framework to analyze 
several important problems in computer science, especially in formal methods
for the analysis of reactive systems.
Graph games are games that proceed for an infinite number of rounds,
where the two players take turns to move a token along the edges of the graph 
to form an infinite sequence of vertices (which is called a {\em play} or 
a {\em trace}). 
The desired set of plays  is described as an $\omega$-regular winning condition.
A {\em strategy} for a player is a recipe that describes how the player 
chooses to move tokens to extend plays,  and a {\em winning} strategy ensures 
the desired set of plays against all strategies of the opponent. 
Some classical examples of graph games in formal methods are as follows:

(a)
If the vertices and edges of a graph represent the states 
and transitions of a reactive system, resp., then the synthesis problem
(Church's problem~\cite{Church62}) asks for the construction of a 
{\em winning strategy} in a graph 
game~\cite{BuchiL69,RamadgeW87,PnueliR89,MadhusudanT01,McNaughton93}.

(b)
The problems of (i)~verification of a branching-time property of 
a reactive system~\cite{EmersonJ91}, where one player models the existential 
quantifiers and the opponent models the universal quantifiers; as well as 
(ii)~verification of open systems~\cite{AlurHK02}, where one player represents the 
controller and the opponent represents the environment; are naturally modeled 
as graph games, where the winning strategies represent the choices of the 
existential player and the controller, respectively.

Moreover, game-theoretic formulations have been used for 
refinement~\cite{HenzingerKR02}, compatibility 
checking~\cite{InterfaceTheories} of reactive systems, 
program repair~\cite{JobstmannGB05}, and synthesis of programs~\cite{CernyCHRS11}.    
Graph games with {\em parity winning} conditions are particularly important 
since all $\omega$-regular winning conditions (such as safety, reachability, 
liveness, fairness) as well as all Linear-time Temporal Logic (LTL) winning 
conditions can be translated to parity conditions~\cite{Safra88,Safra89},
and parity games are equivalent to modal $\mu$-calculus model 
checking~\cite{EmersonJ91}.
In a parity winning condition, every vertex is assigned a non-negative
integer priority from $\set{0,1,\ldots,\numprio-1}$, and a play is winning if 
the highest priority visited infinitely often is even.
Graph games with parity conditions can model all the applications mentioned 
above, and there is a rich literature on the algorithmic study of finite-state 
parity games~\cite{EmersonJ91,BrowneCJLM97,Seidl96,Jurdzinski00,
VogeJ00,JurdzinskiPZ08,Schewe17}.

\paragraph{Explicit vs. symbolic algorithms.}
The algorithms for parity games can be classified broadly 
as {\em explicit} algorithms, where the algorithms operate on the explicit 
representation of the graph game, and {\em implicit or symbolic} algorithms, 
where the algorithms only use a set of predefined operations and do not 
explicitly access the graph game.
Symbolic algorithms are of great significance for the following reasons:
(a)~first, symbolic algorithms are required for large finite-state systems
that can be succinctly represented implicitly (e.g., programs with Boolean 
variables) and symbolic algorithms are scalable, whereas explicit algorithms 
do not scale; and 
(b)~second, for infinite-state systems (e.g., real-time systems modeled as
timed automata, or hybrid systems, or programs with integer domains) only 
symbolic algorithms are applicable, rather than explicit algorithms. 
Hence for the analysis of large systems or infinite-state systems symbolic 
algorithms are necessary.

\paragraph{Significance of set-based symbolic algorithms.}
The most significant class of symbolic algorithms for parity games 
are based on {\em set operations}, where the allowed symbolic operations are:
(a)~basic set operations such as union, intersection, complement, and 
inclusion; and (b)~one step predecessor (Pre) operations.
Note that the basic set operations (that only involve state variables) are 
much cheaper as compared to the predecessor operations (that involve both variables of the current and of the next state). 
Thus in our analysis we will distinguish between the basic set operations and 
the predecessor operations. We refer to the number of sets stored by a set-based
symbolic algorithm as its space.
The significance of set-based symbolic algorithms is as follows:

(a) First, in several domains of the analysis of both infinite-state systems
(e.g., games over timed automata or hybrid systems) as well as large 
finite-state systems (e.g., programs with many Boolean variables, or bounded
integer variables), the desired model-checking question is specified as 
a $\mu$-calculus formula with the above set operations~\cite{deAlfaroHM01,deAlfaroFHMS03}. 
Thus an algorithm with the above set operations provides a symbolic algorithm 
that is directly applicable to the formal analysis of such systems.

(b)
Second, in other domains such as in program analysis, the one-step 
predecessor operators are routinely used (namely, with the 
weakest-precondition as a predicate transformer).
A symbolic algorithm based only on the above operations thus can
easily be developed on top of the existing implementations.
Moreover, recent work~\cite{BeyeneCPR14} shows how efficient procedures (such as 
constraint-based approaches using SMTs) can be used for the computation of the 
above operations in infinite-state games. 
This highlights that symbolic one-step operations can be applied to a large class of problems.

(c)
Finally, if a symbolic algorithm is described with the above very basic 
set of operations, then any practical improvement to these operations in 
a particular domain would translate to a symbolic algorithm that is faster 
in practice for the respective domain.

Thus the problem is practically relevant, and understanding the symbolic 
complexity of parity games is an interesting and important problem.

\paragraph{Previous results.}
We summarize the main previous results for finite-state game graphs 
with parity conditions.
Consider a parity game with $n$ vertices, $m$ edges, and $\numprio$ priorities 
(which is equivalent to $\mu$-calculus 
model-checking of transitions systems with $n$ states, $m$ transitions, 
and a $\mu$-calculus formula of alternation depth $\numprio$).
In the interest of concise presentation, in the following discussion, 
we ignore denominators in~$\numprio$ 
in the running time bounds, see  Appendix~\ref{sec:algoforparity} and 
Theorems~\ref{thm:pmalg_space_efficient} and~\ref{thm:bigstep} for precise bounds.

\emph{Set-based symbolic algorithms.} 
Recall that we use space to refer to the number of sets stored by a symbolic
algorithm. The basic set-based symbolic algorithm (based on the direct evaluation of
the nested fixed point of the $\mu$-calculus formula) for parity games requires 
$O(n^\numprio)$ symbolic operations and space linear in $\numprio$~\cite{EmersonL86}. 
In a breakthrough result~\cite{BrowneCJLM97}, a new set-based symbolic algorithm was 
presented that requires $O(n^{\numprio/2+1})$ symbolic operations, but also 
requires $O(n^{\numprio/2+1})$ many sets, i.e., exponential space as compared to the 
linear space of the basic algorithm.
A simplification of the result of~\cite{BrowneCJLM97} was presented in~\cite{Seidl96}.

\emph{Explicit algorithms.}
The classical algorithm for parity games requires $O(n^{\numprio-1} m)$ time 
and can be implemented in quasi-linear space~\cite{Zielonka98,McNaughton93}, 
which was then improved to the small-progress 
measure algorithm that requires $O(n^{\numprio/2} m)$ time and space to store
$O(\numprio \cdot n)$ integer counters~\cite{Jurdzinski00}. 
The small-progress measure algorithm, which is an explicit algorithm, 
uses an involved domain of the product of integer priorities and \emph{lift} 
operations (which is a lexicographic $\max$ and $\min$ in the involved domain). 
The algorithm shows that the fixed point of the lift operation computes the 
solution of the parity game. 
The lift operation can be encoded with algebraic binary 
decision diagrams~\cite{BustanKV04} but this does not provide a 
set-based symbolic algorithm. 
Other notable explicit algorithms for parity games are as follows:
(a)~a strategy improvement algorithm~\cite{VogeJ00}, which in the worst-case is 
exponential~\cite{Friedmann09};
(b)~a dominion-based algorithm~\cite{JurdzinskiPZ08}
that requires $n^{O(\sqrt{n})}$ time and a randomized $n^{O(\sqrt{n/\log{n}})}$
algorithm~\cite{BjorklundSV03} (both algorithms are sub-exponential, but inherently explicit algorithms); 
and, combining the small-progress measure and the dominion-based algorithm, 
(c)~an $O(n^{\numprio/3} m)$ time algorithm~\cite{Schewe17} and 
its improvement for dense graphs with $\numprio$ sub-polynomial in $n$ to an 
$O(n^{\numprio/3} n^{4/3})$ time algorithm~\cite{ChatterjeeHL15} (both bounds are simplified). 
A recent breakthrough result~\cite{CaludeJKLS17} shows that parity games 
can be solved in $O(n^{\log \numprio})$ time, i.e., quasi-polynomial time.
Follow-up work \cite{JurdzinskiL17,FearnleyJSSW17}
reduced the space requirements from quasi-polynomial to $O(n \log n \log \numprio)$,
i.e., to quasi-linear, space.

While the above algorithms are specified for finite-state graphs, the symbolic 
algorithms also apply to infinite-state graphs with a finite bi-simulation 
quotient (such as timed-games, or rectangular hybrid games), and then $n$ 
represents the size of the finite quotient. 

\paragraph{Our contributions.}
Our results for 
game graphs with $n$ vertices and parity
objectives with $\numprio$ priorities are as follows. 

(1) First, we present a set-based symbolic algorithm that requires 
$O(n^{\numprio/2+1})$ symbolic operations and linear space
(i.e., a linear number of sets). 
		Thus it matches the symbolic operations bound of~\cite{BrowneCJLM97,Seidl96} 
		and brings the space requirements down to a linear number of sets as in
		the classical algorithm (albeit linear in $n$ and not in $\numprio$).

(2) Second, developing on our first algorithm, 
we present a set-based symbolic algorithm that requires 
$O(n^{\numprio/3+1})$ symbolic operations (simplified bound) and linear space. 
Thus it improves the symbolic operations of~\cite{BrowneCJLM97,Seidl96} while
achieving an exponential improvement in the space requirement.
We also present a modification of our algorithm that requires 
$n^{O(\sqrt{n})}$ symbolic operations and at most linear space. 
This is the first linear-space 
set-based symbolic algorithm that requires at most a sub-exponential number of 
symbolic operations.

In the results above the number of symbolic operations mentioned is 
the number of predecessor operations, and in all
cases the number of required basic set operations
(which are usually cheaper) is at most a factor of $O(n)$ more.
Our main results and comparison with previous
set-based symbolic algorithms are presented in the table below.

\begin{center}
  \begin{tabular}{@{}lcc@{}}
    \toprule
    reference & symbolic operations & space\\
    \midrule
    \cite{EmersonL86,Zielonka98} & $O(n^\numprio)$ & $O(\numprio)$ \\
    \cite{BrowneCJLM97,Seidl96} & $O(n^{\numprio/2+1})$ & $O(n^{\numprio/2+1})$\\
    Thm.~\ref{thm:pmalg_space_efficient} & $\mathbf{O(n^{\numprio/2+1})}$ & $\mathbf{O(n)}$ \\
    Thm.~\ref{thm:bigstep} & $\mathbf{\min\{ n^{O(\sqrt{n})},O(n^{\numprio/3+1})\}}$ & $\mathbf{O(n)}$ \\
    \bottomrule
  \end{tabular}  
\end{center}

\paragraph{Technical contribution.}
We provide a symbolic version of the progress measure algorithm. 
The main challenge is to succinctly encode the numerical domain of the progress
measure as sets. 
More precisely, the challenge is to represent $\Theta(n^{c/2})$ many numerical 
values with $O(n)$ many sets, such that they can still be efficiently processed 
by a set-based symbolic algorithm.
For the sake of efficiency our algorithms consider sets $S_\rank$ storing 
all vertices with progress measure at least~$\rank$.
However, there are $\Theta(n^{c/2})$ many such sets $S_\rank$ and thus, to 
reduce the space requirements to a linear number of sets, we use a succinct 
		representation that encodes all the sets $S_\rank$ with just $O(n)$ many 
sets, such that we can restore a set $S_\rank$ efficiently whenever 
it is processed by the algorithm.

\section{Preliminaries and Previous Results}\label{sec:prelimianries}

\subsection{Parity Games}
\paragraph{Game graphs.}
We consider games on graphs played by two adversarial players, denoted 
by $\pe$ (for even) and $\po$ (for odd).
We use $\pl$ to denote one of the players of $\set{\pe,\po}$ and $\op$ to denote 
its opponent. We denote by $(V, E)$ a directed graph with 
$n=\lvert V \rvert$~vertices and $m=\lvert E \rvert$ edges, where
$V$ is the vertex set and $E$ is the edge set.
A \emph{game graph}  $\game = ((V, E),(\ve, \vo))$ is a directed graph $(V, E)$
with a partition of the vertices into player-$\pe$ vertices $\ve$ 
and player-$\po$ vertices $\vo$.
For a vertex $u\in V$, we write $\Out(u)=\set{v\in V \mid (u,v) \in E}$ 
for the set of successor vertices of~$u$.
As a standard convention (for technical simplicity) we consider 
that every vertex has at least one outgoing edge, i.e., 
$\Out(u)$ is non-empty for all vertices~$u$.

\paragraph{Plays.}
A game is initialized by placing a token on a vertex. Then the two
players form an infinite path, called \emph{play}, in the game graph by 
moving the token along the edges. Whenever the token is on a vertex of~$V_\pl$, 
player~$\pl$ moves the token along one of the outgoing edges of the vertex.
Formally, a \emph{play} is an infinite sequence $\sseq$ of vertices such that
$(v_j,v_{j+1}) \in E$ for all $j \geq 0$.

\paragraph{Parity games.}
A \emph{parity game} $\pgame = (\game, \prio)$ with $\numprio$ priorities 
consists of a game graph $G = ((V, E),(\ve, \vo))$ and a \emph{priority
function} $\prio: V \rightarrow [\numprio]$ that assigns an integer from the set $[\numprio] = \set{0, \ldots, \numprio -1}$ to each vertex
(see Figure~\ref{fig:example2} for an example). 
Player~$\pe$ (resp.\ player~$\po$) wins a play of the parity game if the 
\emph{highest} priority occurring infinitely often in the play is 
even (resp.\ odd).
We denote by $P_i$ the set of vertices with priority $i$, i.e., 
$P_i=\{ v \in V\mid \prio(v)=i \}$. 
Note that if $P_i$ is empty for $0 < i < \numprio-1$, then the priorities $> i$
can be decreased by~2 without changing the parity condition, and when $P_{\numprio-1}$
is empty, we simply have a parity game with a priority less; thus we assume 
w.l.o.g.\
$P_i \ne \emptyset$ for $0 < i < \numprio$.

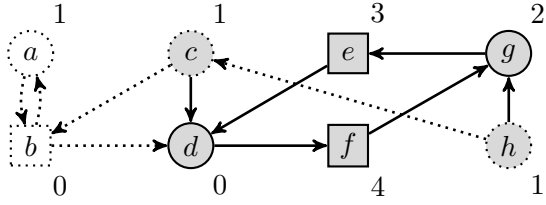
\begin{figure}[t]
\centering
 \begin{tikzpicture}
\matrix[column sep=12mm, row sep=6mm]{
	\node[p1,label=60:$1$, fill=white, dotted] (a) {$a$};
	& \node[p1,label=60:$1$, dotted] (c) {$c$};
	& \node[p2,label=60:$3$] (e) {$e$};
	& \node[p1,label=60:$2$] (g) {$g$};\\
	\node[p2,label=-60:$0$, fill=white, dotted] (b) {$b$};
	& \node[p1,label=-60:$0$] (d) {$d$};
	& \node[p2,label=-60:$4$] (f) {$f$};
	& \node[p1,label=-60:$1$, dotted] (h) {$h$};\\
};
\path[arrow, bend right, dotted] 
	(a) edge (b)
	(b) edge (a);
\path[arrow] 
	(c) edge[dotted] (b)
	(b) edge[dotted] (d)
	(d) edge (f)
	(f) edge (g)
	(g) edge (e)
	(e) edge (d)
	(c) edge (d)
	(h) edge[dotted] (c)
	(h) edge (g)
	;
\end{tikzpicture}
\caption{A parity game with 5 priorities.
Circles denote player~$\pe$ vertices, squares denote player~$\po$ vertices.
The numeric label of a vertex gives its priority,
e.g., $a$ is an $\pe$-vertex with priority~$1$. 
The set of solid vertices is a player-$\pe$ dominion and the union of 
this $\pe$-dominion with the vertices $c$ and $h$ is another $\pe$-dominion
that is equal to the winning set of player $\pe$ in this game. 
The solid edges indicate a winning strategy for player~$\pe$. } 

\label{fig:example2}
\end{figure}

\paragraph{Strategies.}
A \emph{strategy} of a player~$\pl \in \set{\pe,\po}$ is a 
function that, given a finite prefix of a play ending at $v \in V_\pl$, 
selects a vertex from $\Out(v)$ to extend the finite prefix. 
\emph{Memoryless strategies} depend only on the last vertex of the finite prefix.
That is, a memoryless strategy of player~$\pl$ is a function $\straa: V_\pl \rightarrow V$ 
such that for all $v \in V_\pl$ we have $\straa(v) \in \Out(v)$. 
It is well-known that for parity games it is sufficient to consider memoryless 
strategies~\cite{EmersonJ91,McNaughton93}. 
Therefore we only consider memoryless strategies from now on.
A start vertex~$v$, a strategy~$\straa$ for $\pe$, and a 
strategy~$\strab$ for $\po$ describe a unique play $\pat(v,\straa,\strab)=\sseq$,
which is defined as follows: 
$v_0=v$ and for all $i \geq 0$, if $v_i \in \ve$, then $\straa(v_i)=v_{i+1}$, 
and if $v_i \in \vo$, then $\strab(v_i)=v_{i+1}$.

\paragraph{Winning strategies and sets.}
A strategy 
$\straa$ is \emph{winning} for player~$\pe$ at start vertex
$v$ iff for all strategies $\strab$ of player~$\po$ we have that the play 
$\pat(v,\straa,\strab)$ satisfies the parity condition, and analogously for 
winning strategies for player~$\po$.
A vertex~$v$ belongs to the \emph{winning set} $W_\pl$ of player~$\pl$ 
if player~$\pl$ has a winning strategy from start vertex $v$. 
Every vertex is winning for exactly one of the two players.
The algorithmic problem we study for parity games is to compute the winning sets 
of the two players. 
A non-empty set of vertices $D$ is a \emph{player-$\pl$ dominion} 
if player~$\pl$ has a winning strategy from every vertex of $D$ that
also ensures only vertices of $D$ are visited.

\subsection{Set-based Symbolic Operations}\label{sec:symops}
Symbolic algorithms operate on sets of vertices, which are usually 
described by Binary Decision Diagrams (BDD)~\cite{Lee59,Akers78}.
For the symbolic algorithms for parity games we consider the most basic form of
symbolic operations, namely, \emph{set-based symbolic} operations.
More precisely, 
we only allow the following operations:

\noindent \emph{Basic set operations.} First, we allow basic set operations 
like $\cup$, $\cap$, $\setminus$, $\subseteq$, and $=$.

\noindent \emph{One-step operations.} 
Second, we allow the following symbolic one-step operations:\\
(a)~the one-step predecessor operator 
$
  \pre(B)=\{v \in V \mid \exists u \in B: (v,u) \in E\};
$
and 
(b)~the one-step \emph{controllable} predecessor operator
$
  \cpre{\pl}(B) = \left\{ v \in V_\pl \mid \Out(v) \cap B \ne \emptyset \right\} \cup  
		  \left\{ v \in V_\op \mid \Out(v) \subseteq B \right\};
$
i.e., the $\cpre{\pl}$ operator computes all vertices from which $\pl$ can ensure 
that in the next step the successor belongs to the given set~$B$.
Moreover, the $\cpre{\pl}$ operator can be defined using the $\pre$ operator
and basic set operations as follows:
$
  \cpre{\pl}(B)= \pre(B) \setminus (V_\op \cap \pre(V \setminus B)).
$

Algorithms that use only the above operations are called \emph{set-based symbolic}
algorithms. Additionally, successor operations can be allowed but are not needed
for our algorithms. The above symbolic operations correspond to  
primitive operations in 
standard symbolic 
packages like \textsc{CuDD}~\cite{Somenzi15}.

Typically, the basic set operations are cheaper (as they encode relationships
between state variables) as compared to the one-step symbolic operations (which encode 
the transitions and thus the relationship between the variables of the present and of 
the next state).
Thus in our analysis we distinguish between these two types of operations.

For the \emph{space} requirements of set-based symbolic algorithms, 
as per standard convention~\cite{BrowneCJLM97,BustanKV04},
we consider that a set is stored in constant space (e.g., a set can be represented symbolically as one BDD~\cite{Bryant86}). 
We thus consider the space requirement of a symbolic algorithm to be the maximal number 
of sets that the algorithm has to store.

\subsection{Progress Measure Algorithm}\label{sec:pmdef}

We first provide basic intuition for the progress measure~\cite{Jurdzinski00}
and then provide the formal definitions. 
Solving parity games can be reduced to computing the progress 
measure~\cite{Jurdzinski00}.
In Section~\ref{sec:pmalg} we present a set-based symbolic algorithm
to compute the progress measure. 

\subsubsection{High-level intuition}
Towards a high-level intuition behind the progress measure, consider an 
$\pe$-dominion~$D$, i.e., 
player~$\pe$ wins on all vertices of $D$ without leaving $D$.
Fix a play started at a vertex $u \in D$ in which player~$\pe$ follows her 
winning strategy on $D$. In the play from some point on the highest priority 
visited by the play, say~$\prio^*$, has to be even. 
Let $v_*$ be the vertex after which the highest visited priority is $\prio^*$
(recall that memoryless strategies are sufficient for parity games).
Before $v_*$ is visited, the play might have visited vertices with odd priority
higher than~$\prio^*$ but the number of these vertices has to be less than~$n$.
The \emph{progress measure} 
is based on a so-called \emph{lexicographic ranking} function that assigns a \emph{rank} 
to each vertex~$v$, where the rank is a ``vector of counters''  for
the number of times player~$\po$ can force a play to visit an odd priority 
vertex before a vertex with higher even priority is reached. 
If player~$\po$ can ensure a counter value of at least $n$,
then she can ensure that a cycle with highest priority odd is reached from~$v$ 
and therefore player~$\pe$ cannot win from the vertex~$v$.
Conversely, if player~$\po$ can reach a cycle with highest priority odd 
before reaching a higher even priority, then she can also force a play to visit 
an odd priority $n$ times (thus a counter value of $n$) 
before reaching a higher even priority.
In other words, a vertex~$u$ is in the $\pe$-dominion~$D$
if and only if player~$\po$ \emph{cannot} force any counter value to reach~$n$ from~$u$.
When a vertex~$u$ is classified as winning for player~$\po$, it is marked with 
the rank $\top$ and whenever $\po$ has a strategy for some vertex $v$ to 
reach a $\top$-ranked vertex, it is also winning for player~$\po$ and thus ranked $\top$.
Computing the progress measure is done by updating the rank of a vertex 
according to the ranks of its successors and is equal to computing the least
simultaneous fixed point for all vertices with respect to ``ranking functions''.

\paragraph{Strategies from progress measure.}
An additional property of the progress measure
is that the ranks assigned to the vertices of the $\pe$-dominion provide a 
certificate for a winning strategy of player~$\pe$ within the dominion, namely, 
player~$\pe$ can follow edges that lead to vertices with ``lower or equal'' rank
with respect to a specific ordering of the ranks.

We next provide formal definitions of \emph{rank}, the \emph{ranking function},
the ordering on the ranks, the \emph{lift}-operators, 
and finally the \emph{progress measure} (see also~\cite{Jurdzinski00}).

\subsubsection{Formal Definitions}
\paragraph{The progress measure domain $M^\infty_G$.}
We consider parity games with $n$ vertices and priorities~$[\numprio]$.
Let $n_i$ be the number of vertices with priority~$i$ for odd~$i$ 
(i.e., $n_i= |P_i|$), let $n_i= 0$ for even $i$, and let $N_i=[n_i+1]$ for 
$0 \leq i < \numprio$.
Let $M_\game=(N_0 \times N_1 \times \cdots \times N_{\numprio-2} \times N_{\numprio-1})$
be the product domain where every even 
index is~0 and every odd index $i$ is a number between $0$ and $n_i$.
The \emph{progress measure domain} is $M^\infty_\game=M_\game \cup \{\top\}$, where $\top$ 
is a special element called the top element.
Then we have $|M^\infty_\game|=1+\prod_{i=1}^{\lfloor\numprio/2\rfloor}(n_{2i-1}+1) = 
O\big(\big(\frac{n}{\lfloor \numprio/2\rfloor}\big)^{\lfloor \numprio/2\rfloor}\big)$~\cite{Jurdzinski00} 
(this bound uses that w.l.o.g.\
$|P_i| >0$ for each priority $i > 0$).

\paragraph{Ranking functions.}
A \emph{ranking function} $\rho: V \rightarrow M^\infty_\game$ 
assigns to each vertex a \emph{rank}~$\rank$ that is 
either one of the $\numprio$ dimensional vectors in 
$M_\game$ or the top element~$\top$. Note that a rank has at most
$\lfloor \numprio / 2 \rfloor$ non-zero entries.
Informally, we call the entries of a rank with an odd index~$i$ 
a ``counter'' because as long as the top element is not reached,
it counts (with ``carry'', i.e., if $n_i$ is reached, the next highest counter 
is increased by one and the counter at index~$i$ is reset to zero)
the number of times a vertex of priority~$i$ is reached before a vertex 
of higher priority is reached (from some specific start vertex).
The co-domain of $\rho$ is $M^\infty_\game=M_\game \cup \{\top\}$ and 
we index the elements of the vectors from $0$ to $\numprio-1$. 

\paragraph{Lexicographic comparison operator $<$.}
We use the following ordering $<$ of the ranks assigned by $\rho$:
the vectors are considered in the lexicographical order, 
where the left most entry is the least significant one and
the right most entry is the most significant one, and $\top$ is the maximum element
of the ordering.
We write $\bar{0}$ to refer to the all zero vector 
(i.e., the minimal element of the ordering)
and $\bar{N}$ to refer to the maximal vector $(n_0, n_1, \dots, n_{\numprio - 1})$
(i.e., the second largest element, after $\top$, in the ordering).

\paragraph{Lexicographic increment and decrement operations.}
Given a rank~$r$, i.e., either a vector or $\top$, 
we refer to the successor in the ordering~$<$ by $\incr(\rank)$ (with 
$\incr(\top)=\top$), and to the predecessor in the ordering~$<$ by $\decr(\rank)$ 
(with $\decr(\bar{0})=\bar{0}$).
We also consider restrictions of  $\incr$ and $\decr$ to fewer dimensions,
which are described below.
Given a vector $x = (x_0,x_1,x_2, \dots, x_{\numprio-1})$, we denote
by $\proj{x}_\ell$ (for $0 \leq \ell < \numprio$) the vector
$(0, 0, \dots , 0, x_\ell, \dots, x_{\numprio - 1})$, where we set all elements with 
index less than $\ell$ to $0$; in particular $x = \proj{x}_0$.
Intuitively, we use the notation $\proj{x}_\ell$ to ``reset the counters''
for priorities lower than~$\ell$ when a vertex of priority~$\ell$ is reached
(as long as we have not counted up to the top element).
Moreover, we also generalize the ordering to a family of orderings $<_\ell$ where $x <_\ell y$ for two vectors $x$ and $y$ iff $\proj{x}_\ell < \proj{y}_\ell$;
the top element $\top$ is the maximum element of each ordering.
In particular, $x <_0 y$ iff $x < y$ and in our setting also $x <_1 y$ iff $x < y$.
We further have restricted versions $\incr_\ell$ and $\decr_\ell$ of  $\incr$ and 
$\decr$; note that $\decr_\ell$ is a partial function and that
$\ell$ will be the priority of the vertex~$v$ for which we want to 
update its rank and $x$ will be the rank of one of its neighbors in the game graph.
\begin{itemize}
 \item $\incr_\ell(x)$:   For $x=\top$ we have $\incr_\ell(\top)=\top$; 
			  Otherwise $\incr_\ell(x)= \proj{x}_\ell$ if $\ell$ is even and
			  $\incr_\ell(x)=\min\{y \in M^\infty_\game \mid y >_\ell x\}$ if $\ell$ is odd.
 \item $\decr_\ell(x):$ $\decr_\ell(x)=\bar{0}$ if $\proj{x}_\ell=\bar{0}$; 
			Otherwise if $\proj{x}_\ell > \bar{0}$ then 
			$\decr_\ell(x)=\min\{y \in M_\game \mid x=\incr_\ell(y)\}$.
\end{itemize}
For $\bar{0} < \proj{x}_\ell < \top$ we have $\incr_\ell(\decr_\ell(x)) = \decr_\ell(\incr_\ell(x))=\proj{x}_\ell$ while for $\top$ we 
only have $\incr_\ell(\decr_\ell(\top))=\top$ and for $\proj{x}_\ell=\bar{0}$ only $\decr_\ell(\incr_\ell(x))=\bar{0}$. 
By the restriction of $\incr$ by the priority~$\ell$ of~$v$,
for both even and odd priorities the counters for lower (odd) priorities
are reset to zero as long as the top element is not reached. For an odd~$\ell$
additionally the counter for $\ell$ is increased or, if the counter for $\ell$ 
has already been at~$n_\ell$, then one of the higher counters is increased while the 
counter for $\ell$ is reset to zero as well; if no higher counter can be increased 
any more, then the rank of~$v$ is set to~$\top$. 

\paragraph{The $\best$ operation.}
Recall the interpretation of the progress measure as a witness for a 
player-$\pe$ winning strategy on an $\pe$-dominion, where player~$\pe$ wants 
to follow a path of non-increasing rank. The function~$\best$ we define next
reflects the ability of player~$\pe$ to choose the edge leading to the lowest 
rank when he owns the vertex, while for player-$\po$ vertices all edges 
need to lead to non-increasing ranks if player~$\pe$ can win from this vertex.
The function~$\best$ for each vertex~$v$ and ranking function~$\rho$ is given by
\[ 
\best(\rho,v) = 
 \begin{cases}
  \min\{\rho(w)\mid (v,w) \in E\} & \text{if } v \in V_\pe\,, \\
  \max\{\rho(w)\mid (v,w) \in E\} & \text{if } v \in V_\po\,.
 \end{cases}
\]

\paragraph{The $\lift$ operation and the progress measure.}
Finally, the lift operation implements the incrementing of the rank of a vertex~$v$
according to its priority and the ranks of its neighbors:
\[ 
\lift(\rho,v)(u) = 
 \begin{cases}
  \incr_{\prio(v)}(\best(\rho,v)) & \text{if } u=v\,, \\
  \rho(u) & \text{otherwise}\,.
 \end{cases}
\]
The $\lift(.,v)$-operators are monotone and the \emph{progress measure} 
for a parity game is defined as the \emph{least simultaneous fixed point of all $\lift(.,v)$-operators}.
The progress measure can be computed by starting with the ranking function equal 
to the all-zero function and iteratively applying the $\lift(.,v)$-operators
in an arbitrary order~\cite{Jurdzinski00}.
Note that in this case the $\lift(.,v)$-operator assigns only rank vectors~$\rank$
with $\rank = \proj{\rank}_{\prio(v)}$ to~$v$.
See \cite{Jurdzinski00} for a worst-case example for any lifting algorithm.
By \cite{Jurdzinski00}, the winning set of player $\pe$ can be obtained from 
the progress measure by selecting those vertices whose rank is a vector, 
i.e., smaller than $\top$.

\begin{lemma}\cite{Jurdzinski00}\label{lem:jur:pmcorrect}
For a given parity game and the progress measure~$\rho$ with co-domain $M^\infty_\game$, 
the set of vertices with $\rho(v)<\top$ is exactly the winning set of player~$\pe$.
\end{lemma}

This implies that to solve parity games it is sufficient to provide an algorithm 
that computes the least simultaneous fixed point of all $\lift(.,v)$-operators.
The $\lift$ operation can be computed explicitly in $O(m)$ time, which gives
the \textsc{SmallProgressMeasure} algorithm of~\cite{Jurdzinski00}. 
The \textsc{SmallProgressMeasure} algorithm is an explicit algorithm that requires $O(m \cdot |M^\infty_\game|) =
O\big(m \cdot \big(\frac{n}{\lfloor \numprio/2\rfloor}\big)^{\lfloor \numprio/2\rfloor}\big)$  
time and $O(n \cdot \numprio)$ space (assuming constant size integers).

\begin{example}\label{example:parityExplicit}
In this example we apply the explicit progress measure algorithm to the parity game in 
Figure~\ref{fig:example}. We have $n_1=3$ and $n_3=1$ and thus we have to consider ranks in the co-domain 
$M^\infty_G=\{(0,0),(1,0),(2,0),(3,0),(0,1),$ $(1,1),(2,1),(3,1),\top\}$. 
We initialize $\rho(v)=(0,0)$ for all~$v$ and then apply lifting operations as follows:
\begin{enumerate}
  \item First, consider the vertices $a,h,c$. 
	For $v \in \{a,h,c\}$ we have $\best(\rho,v)=(0,0)$ and $\prio(v)=1$ and thus
	the lift operation $\lift(\rho,v)$ sets $\rho(v)=(1,0)$.
	Now consider the vertex $e$ which also hast $\best(\rho,e)=(0,0)$ but $\prio(e)=3$.
	Here $\lift(\rho,a)$ sets $\rho(a)=(0,1)$.
	
  \item Now consider vertices with even priority.
	The vertex $g$ has $\best(\rho,g)=(0,1)$ and as $\prio(g)<3$ we lift $\rho(g)$ to $(0,1)$.
	The vertex $b$ has $\best(\rho,b)=(1,0)$ and as $\prio(g)<1$ we lift $\rho(b)$ to $(1,0)$.
	The vertex $d$ has $\best(\rho,d)=(0,0)$ and $\rho(d)$ stays unchanged.
	Finally the vertex $f$ now has $\best(\rho,f)=(0,1)$ but $\rho(f)$ stays unchanged because $\prio(f)>3$.
	
  \item Now as $g$ was lifted we have $\best(\rho,h)=(1,0)$ and $\rho(h)$ is lifted to $(2,0)$.
	Similarly as $b$ was lifted we have $\best(\rho,a)=(1,0)$ and $\rho(a)$ is lifted to $(2,0)$, which in turn causes
	that $\best(\rho,b)=(2,0)$ and $\rho(b)$ being lifted to $(2,0)$. This behavior continues until both
	$\rho(a)$ and $\rho(b)$ are lifted to $\top$.
	
  \item Now we already have a fixed-point with $\rho(f)=\rho(d)=(0,0)$, 
	$\rho(c)=(1,0)$,
	$\rho(h)=(2,0)$,
	$\rho(e)=\rho(g)=(0,1)$, and
	$\rho(a)=\rho(b)=\top$.
\end{enumerate}\smallskip
That is, by Lemma~\ref{lem:jur:pmcorrect}, the winning set of player~$\pe$ is $\{c,d,e,f,g,h\}$. \qee 
\end{example}

\section{Set-based Symbolic Progress Measure Algorithm for Parity Games}
\label{sec:pmalg}
In this section we present a set-based symbolic algorithm for 
parity games, with $n$ vertices and $\numprio$ priorities, by showing 
how to compute a progress measure (see Section~\ref{sec:pmdef})
using only set-based symbolic operations (see Section~\ref{sec:symops}). 
In Appendix~\ref{sec:par5} we provide additional intuition for
our algorithm for the special case of $5$ priorities.
All proofs are in Appendix~\ref{app:pargen}.

\paragraph{Key differences and challenges.}
We mention the key differences of Algorithm~\ref{alg:SymbolicParityDominion} and 
the explicit progress-measure algorithm (\cite{Jurdzinski00}, see
Section~\ref{sec:pmdef}).

(1) The main challenge for an efficient set-based symbolic algorithm 
similar to the \textsc{SmallProgressMeasure} algorithm 
is to represent $\Theta(n^{c/2})$ many numerical 
values succinctly with $O(n)$ many sets, such that they can still be efficiently 
processed by a symbolic algorithm.

(2) To exploit the power of symbolic operations, in each iteration of the algorithm 
we compute all vertices whose rank can be increased to a certain value $\rank$.
This is in sharp contrast to the explicit progress-measure 
algorithm, where vertices are considered one by one 
and the rank is increased to the maximal possible value. 

\paragraph{Key concepts.}
Recall 
that the progress measure for parity games
is defined as the least simultaneous fixed point of the $\lift(\rho,v)$-operators 
on a ranking function~$\rho: V \rightarrow M^\infty_\game$. There are two key aspects 
of our algorithm:

(1) \emph{Symbolic encoding of numerical domain.}
In our symbolic algorithm we cannot directly deal with the ranking function
but have to use sets of vertices to encode it. 
We first formulate our algorithm with sets $S_\rank$ for $\rank \in  M^\infty_\game$
that contain all vertices that have rank $\rank$ or higher;
that is, given a function~$\rho$, the corresponding sets 
are $S_\rank=\{v \mid \rho(v) \geq r\}$. 
On the other hand, given a family of sets $\{S_\rank\}_\rank$,
the corresponding ranking function $\rho_{\{S_\rank\}_\rank}$ is given by
$
  \rho_{\{S_\rank\}_\rank}(v) = \max \{r \in M^\infty_\game \mid v \in S_\rank \}
$.
This formulation encodes the numerical domain with sets but uses exponential in $\numprio$ many sets.

(2) \emph{Space efficiency.} 
We refine the algorithm to directly encode the ranks with one set 
for each possible index-value pair. This reduces the required 
number of sets to linear at the cost of increasing the number of set 
operations only by a factor of $n$;
the number of one-step symbolic operations does not increase.

We first present the variant that uses an exponential number of sets and then
show how to reduce the number of sets to linear.

The above ideas yield a set-based symbolic algorithm, but since we now deal with 
sets of vertices, as compared to individual vertices, the correctness needs 
to be established.
The non-trivial aspect of the proof is to identify appropriate 
\emph{invariants on sets} (which we call \emph{symbolic invariants}, 
see Invariant~\ref{invariant}) and use them to establish the correctness.

\subsection{The Set-based Symbolic Progress Measure Algorithm}

\paragraph{The codomain $M^\infty_\domsize$.} We formulate our 
algorithm such that it cannot only compute the winning sets of the players but 
also $\pe$-dominions of size at most $\domsize + 1$. (For $\po$-dominions add one to 
each priority and exchange the roles of the two players.) The only change needed
 for this is to use the codomain $M^\infty_\domsize$, instead of $M^\infty_\game$, 
 for the $\incr$ and $\decr$ operations. The codomain $M^\infty_\domsize$
 contains all ranks of $M^\infty_\game$ whose entries sum up to at most $\domsize$ 
 (see Appendix~\ref{sec:bigstepexpl}).

\begin{algorithm}
	\small
	\SetAlgoRefName{SymbolicParityDominion}
	\caption{Symbolic Progress Measure Algorithm}
	\label{alg:SymbolicParityDominion}
	\SetKwInOut{Input}{Input}
	\SetKwInOut{Output}{Output}
	\SetKw{break}{break}
	\BlankLine
	\Input{%
	  \emph{parity game} $\pgame = (\game, \prio)$, with 
	  \emph{game graph} $\game = ((V, E),(\ve, \vo))$, \\
	  \emph{priority function} $\prio: V \rightarrow [\numprio]$, and
	  \emph{parameter} $\domsize \in [0, n] \cap \mathbb{N}$
	}
	\Output
	{
	  Set containing all $\pe$-dominions of size $\le \domsize+1$, which is an $\pe$-dominion or empty.
	}
	\BlankLine
	$S_{\bar{0}} \gets V$;\
	${S_\rank} \gets {\emptyset}$ for $\rank \in M^\infty_\domsize \setminus \{\bar{0}\}$\;
	$\rank \gets \incr(\bar{0})$\;
	\While
	  {
	    $true$
	  }
	  {
	    \If{$\rank \not= \top$\label{alg:startupdateS}}
	    {
	      Let $\ell$ be maximal such that $\rank=\proj{\rank}_\ell$\; 
	      $S_\rank \gets S_\rank \cup 
				\bigcup_{1 \le \idx \leq (\ell+1)/2}
				\left(
				  \cpre{\po}(S_{{\decr_{2\idx-1}(\rank)}}) \cap P_{2\idx-1} 
				\right)$\;\label{alg:symbParityc_line6}
		\Repeat
		{
		  a fixed-point for $S_\rank$ is reached
		}
		{
		  $S_\rank \gets S_\rank \cup \left( \cpre{\po}(S_\rank) \setminus \bigcup_{\ell < \idx < \numprio} P_\idx \right)$\label{alg:symbParityc_line8}
		}\label{alg:endupdateS}  
	    }
	    \ElseIf{$\rank=\top$\label{alg:startupdatetop}}
	    {
		$S_{\top} \gets  S_\top \cup 
				    \bigcup_{1 \le \idx \leq \lfloor \numprio/2 \rfloor}
				    \left(
				      \cpre{\po}(S_{{\decr_{2\idx-1}(\top)}}) \cap P_{2\idx-1} 
				    \right)$\;\label{alg:symbParityc_line11}
		\Repeat
		{
		  a fixed-point for $S_\top$ is reached
		}
		{
		  $S_\top \gets S_\top \cup \left( \cpre{\po}(S_\top) \right)$\label{alg:symbParityc_line13}
		}\label{alg:endupdatetop}
	    }
	    $\rank'\gets \decr(\rank)$\;\label{alg:symbParityc_line15}
	    \If{$S_{\rank'} \supseteq  S_\rank$ and $\rank<\top$\label{alg:symbParityc_line16}}
	    {
	      $\rank \gets \incr(\rank)$\label{alg:rankinc}
	    }
	    \ElseIf{$S_{\rank'} \supseteq  S_\rank$ and $\rank=\top$
	    \label{alg:ifterminates}}
	    {
	      \break\label{alg:terminates}
	    }
	    \Else
	    {\label{alg:symbParityc_line20}
	      \Repeat{$S_{\rank'} \supseteq  S_\rank$\label{alg:symbParityc_line24}}
	      {\label{alg:rankloop}
		  $S_{\rank'} \gets S_{\rank'} \cup S_\rank$\;\label{alg:symbParityc_line22}
		  $\rank'\gets \decr(\rank');$
	      }
	      $\rank \gets \incr(\rank')$\label{alg:rankafterloop}\;\label{alg:symbParityc_line25}
	    }
	  }
	 \Return $V \setminus S_\top$
\end{algorithm}
\paragraph{The sets $S_\rank$ and the ranking 
function~$\rho_{\{S_\rank\}_\rank}$.}
The algorithm implicitly maintains a rank for each vertex. 
A vertex is contained in a set $S_\rank$
only if its maintained rank is \emph{at least}~$\rank$. 
Each set $S_\rank$ is monotonically increasing throughout the algorithm. 
The rank of a vertex~$v$ is the highest $\rank$ such that $v \in S_\rank$.
In other words, the family of sets~$\{S_\rank\}_\rank$ defines the ranking function
$\rho_{\{S_\rank\}_\rank}(v) = \max \{\rank \in M^\infty_h \mid v \in S_\rank \}$.
When the rank of a vertex is increased, this information has to be propagated to its 
predecessors. This is achieved efficiently by maintaining \emph{anti-monotonicity}
among the sets, i.e., we have $S_{\rank'} \supseteq S_\rank$ for all $\rank$ and 
all $\rank' < \rank$ before and after each iteration.
Anti-monotonicity together with defining the sets $S_{\rank'}$ to contain 
vertices with rank \emph{at least} $\rank'$ instead of \emph{exactly} 
$\rank'$ enables us to decide whether the rank of a vertex~$v$ can be increased 
to $\rank$ by only considering one set $S_{\rank'}$.

\paragraph{Structure of the algorithm.}
The set $S_{\bar{0}}$ is initialized with the set of all vertices~$V$,
while all other sets $S_\rank$ for $\rank > \bar{0}$ are initially empty, i.e., the 
ranks of all vertices are initialized with the zero vector. The variable~$\rank$
is initially set to the second lowest rank $\incr(\bar{0})$ that is one
at index 1 and zero otherwise. In the while-loop the set $S_\rank$ is updated
for the value of $\rank$ at the beginning of the iteration (see below). 
After the update of $S_\rank$, it is checked whether the set corresponding 
to the next lowest rank already contains the vertices newly added to $S_\rank$,
i.e., whether the anti-monotonicity is preserved. 
If the anti-monotonicity is preserved despite the update of $S_\rank$,
then for $\rank < \top$ the value of $\rank$ is increased to the next highest rank
and for $\rank = \top$ the algorithm 
terminates. Otherwise
the vertices newly added to $S_\rank$ are also 
added to all sets with $\rank' < \rank$ that do not already contain them; 
the variable $\rank$ is then updated to the lowest
$\rank'$ for which a new vertex is added to $S_{\rank'}$ in this 
iteration.

\paragraph{Update of set~$S_\rank$.}
To reach a simultaneous fixed point of the lift-operators, the rank of a vertex~$v$
has to be increased to $\lift(\rho_{\{S_\rank\}_\rank}, v)(v)$
whenever the value of $\lift(\rho_{\{S_\rank\}_\rank}, v)(v)$ is strictly higher 
than $\rho_{\{S_\rank\}_\rank}(v)$ for the current ranking function
$\rho_{\{S_\rank\}_\rank}$. 
Now consider a fixed iteration of the while-loop and let $\rank$ 
be as at the beginning of the while-loop.
Let $\rho_{\{S_\rank\}_\rank}$ be denoted by $\rho$ for short.
In this update of the set $S_\rank$ we want to add to $S_\rank$
all vertices~$v$ with $\rho(v) < \rank$ and $\lift(\rho, v)(v) \ge \rank$
under the condition that the priority of~$v$ allows $v$ to be assigned the 
rank~$\rank$, i.e., $\rank = \proj{\rank}_{\prio(v)}$.
Note that by the anti-monotonicity property the set $S_\rank$ already contains 
all vertices with $\rho(v) \ge \rank$.

(1)
We first consider the case 
$\rank < \top$.
Let~$\ell$ be maximal such that $\rank=\proj{\rank}_\ell$,
i.e., the first $\ell$ entries with indices $0$ to $\ell - 1$ of $\rank$
are $0$ and the entry with index $\ell$ is larger than $0$. Note that 
$\ell$ is odd. 
We have that only the $\lift(.,v)$-operators with $\prio(v)\leq \ell$ can
increase the rank of a vertex to $\rank$
as all the others would set the element with index~$\ell$ to $0$.

Recall that $\lift(\rho,v)(v) = \incr_{\prio(v)}(\best(\rho,v))$.
The function $\best$ is implemented by the $\cpre{\po}$ operator: For a player-$\pe$
vertex the value of $\best$ increases only if the ranks of all 
successor have increased, for a player-$\po$ vertex it increases as 
soon as the maximum rank among the successor vertices has increased.
The function $\incr_{\prio(v)}(x)$ for $x < \top$
behaves differently for odd and even $\prio(v)$ (see Section~\ref{sec:pmdef}):
If $\prio(v)$ is odd, then $\incr_{\prio(v)}(x)$ is the smallest
rank~$y$ in $M^\infty_\domsize$ such that $y >_{\prio(v)} x$, i.e.,
$y$ is larger than $x$ w.r.t.\ indices $\ge \prio(v)$.
If $\prio(v)$ is even, then $\incr_{\prio(v)}(x)$ is equal to $x$ with 
the indices lower than $\prio(v)$ set to 0.

(i)
First, consider a $\lift(\rho,v)$ operation with odd $\prio(v) \le \ell$, i.e., 
let $\prio(v)=2\idx-1$ for some $1 \le \idx \le (\ell+1)/2$.
Then $\lift(\rho, v)(v) \ge \rank$ only if 
(a) $v \in V_\pe$ and all successors~$w$ have 
$\rho(w) \geq  \decr_{2\idx-1}(\rank)$, or 
(b) $v \in V_\po$  and one successor~$w$ has $\rho(w) \geq  \decr_{2\idx-1}(\rank)$.
That is, $\lift(\rho, v)(v) \ge \rank$ only if $v \in 
\cpre{\po}(S_{\decr_{2\idx-1}(\rank)})$.
Vice versa, we have that if $v \in \cpre{\po}(S_{\decr_{2\idx-1}(\rank)})$ then 
by $\rho = \rho_{\{S_\rank\}_\rank}$ also $\lift(\rho, v)(v) \ge \rank$.
This observation is implemented in \ref{alg:SymbolicParityDominion} in 
line~\ref{alg:symbParityc_line6}, where such vertices $v$ are added to $S_\rank$. 

(ii)
Now, consider a $\lift(\rho,v)$ operation with even $\prio(v) \le \ell$,
i.e., let $\prio(v)=2\idx$ for some $1 \le \idx \le \ell/2$.
Then $\lift(\rho, v)(v) \ge \rank$ only if 
(a) $v \in V_\pe$ and all successors $w$ have $\rho(w) \geq  \rank$, or 
(b) $v \in V_\po$  and one successor $w$ has $\rho(w) \geq  \rank$.
That is, $\lift(\rho, v)(v) \ge \rank$ only if $v \in \cpre{\po}(S_{\rank})$.
Vice versa, we have that if $v \in \cpre{\po}(S_{\rank})$ then 
$\lift(\rho, v)(v) \ge \rank$.
In \ref{alg:SymbolicParityDominion} these vertices are added iteratively in 
line~\ref{alg:symbParityc_line8} until a fixed point is reached.
The algorithm 
also adds vertices~$v$
with odd priority to $S_\rank$, but due do the above argument 
we have $\lift(\rho, v)(v) > \rank$ and thus they can be included in $S_\rank$.

(2)
The case $\rank = \top$ 
works 
similarly except that (a)~every
vertex is a possible candidate for being assigned the rank $\top$, 
independent of its priority (line~\ref{alg:symbParityc_line11}),
and (b)~whenever $x$ is equal to $\top$, $\incr_{\prio(v)}(x)$ assigns 
the rank $\top$ independently of 
$\prio(v)$ (line~\ref{alg:symbParityc_line13}).

\paragraph{Sketch of bound on number of symbolic operations.}
Observe that each rank~$\rank$ is considered in at least one iteration of the 
while-loop but is only reconsidered in a later iteration 
if at least one vertex was added to the set $S_\rank$ since the last time $\rank$
was considered; in this case $O(\numprio)$ one-step operations are performed.
Thus the number of symbolic operations per set~$S_\rank$
is of the same order as the number of times a vertex is added to the 
set.
Hence the algorithm can be implemented with $O(\numprio \cdot 
n \cdot |M^\infty_\domsize|)$
symbolic operations. For the co-domain $M^\infty_\game$ the bound 
$O(\numprio \cdot n \cdot |M^\infty_\game|)$ is analogous.

\begin{example}\label{example:paritySymbolic}
In this example we apply Algorithm~\ref{alg:SymbolicParityDominion} to the parity game in 
Figure~\ref{fig:example2}. We have $n_1=3$ and $n_3=1$ and thus we have to consider ranks in the co-domain 
$M^\infty_G=\{(0,0),(1,0),(2,0),(3,0),(0,1),$ $(1,1),(2,1),(3,1),\top\}$ (we ignore entires of ranks that are always zero in this notation). \smallskip

\noindent The algorithm initializes the set $S_{(0,0)}$ to $\{a,b,c,d,e,f,g,h\}$
and $\rank$ to $(1,0)$. 
All the other sets $S_\rank$ are initialized as the empty set.
It then proceeds as follows:
\begin{enumerate}
\item In the first iteration of the while-loop it processes $\rank=(1,0)$.
We have $\ell = 1$ and thus the only possible value of $k$ in line~\ref{alg:symbParityc_line6} is $k = 1$. That is, 
line~\ref{alg:symbParityc_line6} adds the vertices in  $\cpre{\po}(S_{0,0}) \cap P_1 = \{a, c, h\}$ to $S_{(1,0)}$
and then in line~\ref{alg:symbParityc_line8} also $b$ is added.
We obtain $S_{(1,0)}=\{a,b,c,h\}$ and as $S_{(1,0)} \subseteq  S_{(0,0)}$,
the rank $\rank$ is increased to $(2,0)$.

\item In the second iteration it processes $\rank=(2,0)$ and
the vertex $a$ is added to $S_{(2,0)}$ in line~\ref{alg:symbParityc_line6} 
and the vertex $b$ is added to $S_{(2,0)}$ in line~\ref{alg:symbParityc_line8},
i.e., $S_{(2,0)}=\{a,b\}$, and $\rank$ is set to $(3,0)$.

\item When processing $\rank=(3,0)$ the set $S_{(3,0)}$ is updated to 
$\{a,b\}$ and $\rank$ is increased to $(0,1)$.

\item Now the algorithm processes the rank $(0,1)$ the first time. We have 
$\ell = 3$ and thus the possible values for $k$ are $1$ and $2$.
     The vertex $a$ is added to $S_{(0,1)}$ because it is contained in
     $\cpre{\po}(S_{3,0}) \cap P_1$ and the vertex $e$ is added because it is 
     contained in $\cpre{\po}(S_{0,0}) \cap P_3$ in line~\ref{alg:symbParityc_line6}. 
     Finally, also $b$ and $g$ are added in line~\ref{alg:symbParityc_line8}.
     That is, we have  $S_{(0,1)}=\{a,b,e,g\}$.
     Now as $S_{(3,0)}\not\subseteq S_{(0,1)}$, $S_{(2,0)}\not\subseteq S_{(0,1)}$ and $S_{(1,0)}\not\subseteq S_{(0,1)}$,
     we have to decrease $\rank$ to $(1,0)$, and also to modify the other 
     sets with smaller rank as follows: $S_{(1,0)}=\{a,b,c,e,g,h\}$; and
     $S_{(2,0)}=S_{(3,0)}=\{a,b,e,g\}$.
\item The algorithm considers $\rank=(1,0)$ again, makes no changes to $S_{(1,0)}$
      and sets $\rank$ to $(2,0)$.
\item Now considering $\rank=(2,0)$, the vertex $h$ is added to the set $S_{(2,0)}$
      in line~\ref{alg:symbParityc_line6}, i.e.,  $S_{(2,0)}=\{a,b,e,g,h\}$,
      and, as $h$ is already contained in  $S_{(1,0)}$, $\rank$ is increased to $(3,0)$.
\item The set $S_{(3,0)}$ is not changed and $\rank$ is increased to
      $(0,1)$.
\item The set $S_{(0,1)}$ is not changed and $\rank$ is increased to
      $(1,1)$.
\item The vertex $a$ is added to $S_{(1,1)}$ in line~\ref{alg:symbParityc_line6} 
      and the vertex $b$ is added to $S_{(1,1)}$ in line~\ref{alg:symbParityc_line8},
      i.e., $S_{(1,1)}=\{a,b\}$, and $\rank$ is increased to $(2,1)$.
\item The vertices $a,b$ are added to $S_{(2,1)}$, i.e., $S_{(2,1)}=\{a,b\}$,
      and $\rank$ is increased to $(3,1)$.      
\item The vertices $a,b$ are added to $S_{(3,1)}$, i.e., $S_{(3,1)}=\{a,b\}$,
      and $\rank$ is increased to $\top$.      
\item The vertex $a$ is added to $S_{\top}$ in line~\ref{alg:symbParityc_line11} 
      and $b$ is added to $S_{\top}$ in line~\ref{alg:symbParityc_line13},
      i.e., $S_\top=\{a,b\}$.      
      Now as $S_{(3,1)} \subseteq S_\top$, the algorithm terminates.
\end{enumerate}\smallskip

Finally we have that 
$S_{(0,0)} = \{a,b,c,d,e,f,g,h\}$,
$S_{(1,0)} = \{a,b,c,e,g,h\}$,
$S_{(2,0)} = \{a,b,e,g,h\}$,
$S_{(3,0)} = S_{(0,1)} = \{a,b,e,g\}$,
and 
$S_{(1,1)} = S_{(2,1)} = S_{(3,1)} = S_{\top} = \{a,b\}$,
That is, the algorithm returns $\{c,d,e,f,g,h\}$ as the winning set of player~$\pe$.
The final sets of the algorithm corresponds to the progress measure $\rho$ with
$\rho(f)=\rho(d)=(0,0)$, 
$\rho(c)=(1,0)$,
$\rho(h)=(2,0)$,
$\rho(e)=\rho(g)=(0,1)$, and
$\rho(a)=\rho(b)=\top$.
\qee 
\end{example}

\paragraph{Outline correctness proof.}
In the following proof we show that 
when Algorithm~\ref{alg:SymbolicParityDominion} terminates, 
the ranking function $\rho_{\{S_\rank\}_\rank}$ is equal to
the progress measure for the given parity game and the co-domain~$M^\infty_\domsize$. 
The same proof applies to the co-domain $M^\infty_\game$.
The algorithm returns the set of vertices that
are assigned a rank $< \top$ when the algorithm terminates. 
By \cite{Schewe17} 
(see Appendix~\ref{sec:bigstepexpl}) this set is 
 an $\pe$-dominion that contains all $\pe$-dominions of size at most $\domsize + 1$
 when the co-domain $M^\infty_\domsize$ is used, and by Lemma~\ref{lem:jur:pmcorrect}
 this set is equal to the winning set of player~$\pe$ when the co-domain
 $M^\infty_\game$ is used.
 Thus it remains to show that $\rho_{\{S_\rank\}_\rank}$ equals the 
 progress measure for the given co-domain when the algorithm terminates.
 We show that maintaining the following invariants over all iteration  of the
 algorithm is sufficient for this and then prove that the invariants are 
maintained. All proofs are in Appendix~\ref{app:pargen} and are described 
for the co-domain $M^\infty_\domsize$.

\begin{invariant}[Symbolic invariants]\label{invariant}
In Algorithm~\ref{alg:SymbolicParityDominion} the following three invariants hold. 
Every rank is from the 
co-domain $M^\infty_\domsize$ and the $\lift(.,v)$-operators are defined w.r.t.\
the co-domain. Let $\tilde{\rho}$ be the progress measure of the given parity game
 and let $\rho_{\{S_{\rank}\}_{\rank}}(v) = 
\max \{\rank \in M^\infty_h \mid v \in S_{\rank} \}$
be the ranking function with respect to the sets $S_{\rank}$ that are maintained 
by the algorithm. 

	\begin{enumerate}
			\item
			Before and after each iteration of the while-loop we have that if a 
			vertex $v$ is in a set $S_{\rank_1}$ then it is also in $S_{\rank_2}$ 
			for all $\rank_2 < \rank_1$ \upbr{anti-monotonicity}.\label{inv:monotonicity}
			
			\item
			Throughout Algorithm~\ref{alg:SymbolicParityDominion} we have
$\tilde{\rho}(v) \geq \rho_{\{S_{\rank}\}_{\rank}}(v)$ for all $v \in V$.\label{inv:atleast}
			
			\item Before and after each iteration of the while-loop we have for the rank
			stored in~$\rank$ and all vertices~$v$ either 
			$\lift(\rho_{\{S_\rank\}_\rank},v)(v)  \ge  \rank$ or
			$\lift(\rho_{\{S_\rank\}_\rank},v)(v) = \rho_{\{S_\rank\}_\rank}(v)$.
			\upbr{b} After 
			the update of $S_\rank$ and before the update of $\rank$ we additionally have 
			$v \in S_\rank$ for all vertices~$v$ with 
			$\lift(\rho_{\{S_\rank\}_\rank},v)(v) = \rank$ \upbr{closure property}\label{inv:closure}.
	\end{enumerate}
\end{invariant}
\paragraph{Informal description of invariants.}
Invariant~\ref{invariant}\upbr{1} ensures 
that the definition of the sets $S_r$ and the ranking function 
$\rho_{\{S_\rank\}_\rank}$ is sound; 
Invariant~\ref{invariant}\upbr{2}
guarantees that $\rho_{\{S_\rank\}_\rank}$ is a lower bound on $\tilde{\rho}$
throughout the algorithm; and 
Invariant~\ref{invariant}\upbr{3} shows 
that when the algorithm 
terminates, a fixed point of the ranking function $\rho_{\{S_\rank\}_\rank}$
with respect to the $\lift(.,v)$-operators is reached.
Together these three properties guarantee that when the algorithm terminates 
the function $\rho_{\{S_\rank\}_\rank}$ corresponds to the progress measure, 
i.e., to the least simultaneous fixed point of the $\lift(.,v)$-operators.
We prove the invariants by induction over 
the iterations of the while-loop. In particular, Invariant~\ref{invariant}(1) 
is ensured by adding vertices newly added to a set $S_\rank$ also to sets $S_{\rank'}$ 
with $\rank' < \rank$ that do not already contain them at the end of each iteration of 
the while-loop. For Invariant~\ref{invariant}(2) we show that whenever 
$\rho_{\{S_{\rank}\}_{\rank}}(v)$ is increased, i.e., $v$ is added to 
the set $S_\rank$, then no fixed point of the lift-operator for $v$ was reached 
yet and thus also the progress measure for $v$ has to be at least as high as the new 
value of $\rho_{\{S_{\rank}\}_{\rank}}(v)$.
The intuition for the proof of Invariant~\ref{invariant}(3) is as follows: 
We first show that 
$\lift(\rho_{\{S_\rank\}_\rank},v)(v) = \rho_{\{S_\rank\}_\rank}(v)$
remains to hold for all vertices $v$ for which the value of 
$\rho_{\{S_\rank\}_\rank}(v)$ is less than the smallest value $\rank'$ for
which $S_{\rank'}$ was updated in the considered iteration. In iterations in which
the value of the variable $\rank$ is not increased, this is already sufficient 
to show part~(a) of the invariant. If $\rank$ is increased, we additionally
use part~(b) to show part~(a). For part~(b) we prove by case analysis that,
before the update of the variable~$\rank$, a vertex with 
$\lift(\rho_{\{S_\rank\}_\rank},v)(v) = \rank$ is included in $S_\rank$.
The correctness of the algorithm then follows from the invariants as outlined above.

\subsection{Reducing Space to Linear.}\label{sec:space}
Algorithm~\ref{alg:SymbolicParityDominion} requires $\lvert M^\infty_\game \rvert$ 
many sets $S_\rank$, which is drastically beyond the space requirement of the progress
measure algorithm for explicitly represented graphs.
Thus we aim to reduce the space requirement to $O(n)$ many sets
in a way that still allows to restore the sets $S_\rank$ efficiently.
For the sake of readability, we assume for this part that $\numprio$ is even.

\paragraph{Main Idea.}
The main idea to reduce the space requirement is as follows.

  (1)
  Instead of storing sets $S_{\rank}$ corresponding to a specific rank,
  we encode the value of each coordinate of the rank $\rank$ separately.
  That is, we define the sets $C^i_0, C^i_1\dots,C^i_{n_i}$ for each odd 
  priority~$i$. Intuitively, a vertex is in the set $C^i_x$ iff the $i$-th 
  coordinate of the rank of $v$ is~$x$. Given these $O(\numprio + n) \in O(n)$ sets, we have encoded the exact rank
  vector~$\rank$ of each vertex with $\rank < \top$. 
  To also cover vertices with rank $\top$, we additionally store the set 
  $S_\top$.

  (2)
  Whenever the algorithm needs to process a set $S_{\rank}$, we reconstruct 
  it from the stored sets, using a linear number of set operations.
  Algorithm~\ref{alg:SymbolicParityDominion} has to be 
  adapted as follows. 
  First, at the beginning of each iteration we have to compute the set $S_\rank$ 
  and up to $\numprio / 2$ 
  sets $S_{\rank'}$ that correspond to some predecessor $\rank'$ of $\rank$.
  Second, at the end of each iteration we have to update the sets $C^i_x$ to incorporate the updated set $S_\rank$.

\paragraph{Computing a set $S_\rank$ from the sets $C^i_x$.}
Let $r_i$ denote the $i$-th entry of $r$. To obtain the set $S^=_\rank$ of vertices with rank \emph{exactly}~$\rank$ (for $\rank < \top$), one can simply 
compute the intersection
$
  \bigcap_{1 \leq \idx \leq \numprio/2}
      C^{2\idx-1}_{\rank_{2\idx-1}}
$ 
of the corresponding sets $C^i_x$. 
However, in the algorithm we need the sets~$S_\rank$ containing 
all vertices~$v$ with a rank at least $\rank$
and computing all sets $S^=_{\rank'}$ with $\rank' \geq \rank$ is not efficient.
Towards a more efficient method to compute $S_\rank$, recall 
that a rank $\rank'< \top$ is higher than $\rank$ if either 
(a) the right most odd element of $\rank'$ 
is larger than the corresponding element in in $\rank$, i.e., $\rank'_{\numprio-1} >\rank_{\numprio-1}$,  or 
(b) if $\rank$ and $\rank'$ coincide on the $i$ right most odd elements,  
i.e., $\rank'_{\numprio - 2\idx + 1} = \rank_{\numprio - 2\idx + 1}$ for $1 \leq \idx  \le i$,
and $\rank'_{\numprio - 2i -1} > \rank_{\numprio - 2 i - 1}$. 
In case (a) we can compute the corresponding vertices by
$$
  S^0_\rank=\bigcup_{\rank_{\numprio-1} < x \le n_{\numprio-1}} C^{\numprio-1}_x
$$
while in case (b) we can compute the corresponding vertices by
$$
  S^i_\rank=\bigcap_{1 \leq \idx \le i}  C^{\numprio-2\idx+1}_{\rank_{\numprio-2\idx+1}} 
\cap \bigcup_{\rank_{\numprio - 2 i - 1} < x \leq n_{\numprio - 2 i - 1}} 
C^{\numprio - 2 i - 1}_x
$$ 
for $1 \le i \le \numprio / 2 - 1$.
That is, we can reconstruct the set $S_\rank$ by the following union of the above sets $S^i_\rank$,
the set $S^=_\rank$ of vertices with rank $\rank$, and  the set $S_\top$  of vertices with rank $\top$:
$$
    S_\rank= S_\top \cup S^=_\rank \cup \bigcup_{i=0}^{\numprio/2-1} S^i_\rank
$$
Hence, a set $S_\rank$ can be computed with 
$O(\numprio + n) \in O(n)$ many $\cup$ and 
$O(\numprio)$ many $\cap$ operations; 
for the latter bound we use an additional set to store the set
$\bigcap_{1 \leq \idx \leq i} C^{\numprio-2\idx+1}_{\rank_{\numprio-2\idx+1}}$ 
for the current value of $i$, such that for each set $S^i_\rank$ we just 
need two $\cap$ operations. This implies the following lemma.
\begin{lemma}
Given the sets $C^i_x$ as defined above, we can compute the set $S_\rank$ that contains
all vertices with rank at least $\rank$ with $O(n)$ many symbolic set operations.
No symbolic one-step operation is needed.
\end{lemma}

\paragraph{Updating a set $S_\rank$.}
Now consider we have updated a set $S_\rank$ during the iteration of the while-loop,
and now we want to store the updated set $S_\rank$ within the 
sets $C^i_x$. 
That is, 
we have already computed the fixed-point for $S_\rank$ and are now in line~\ref{alg:symbParityc_line15}
of the Algorithm.
To this end, let $S^\text{old}_\rank$ be the set as stored in $C^i_x$ 
and $S^\text{new}_\rank$ the updated set, 
which is a superset of the old one.
First, one  computes the difference $S^{\text{diff}}_\rank$ between the two sets $S^{\text{diff}}_\rank = 
S^{\text{new}}_\rank \setminus S^{\text{old}}_\rank$; intuitively, the set 
$S^\text{diff}_\rank$ contains the vertices for which the algorithm has increased
the rank.
Now for the vertices of $S^\text{diff}_\rank$ we have to (i) delete their old values by
updating $C^i_x$ to $C^i_x \setminus S^\text{diff}_\rank$ 
for each  $i \in \{1, 3, \dots, \numprio-1 \}$ and each $x \in \{0, \dots, n_i\}$
and (ii) store the new values by updating $C^i_x$ to $C^i_x \cup S^\text{diff}_\rank$ for $i \in \{1, 3, \dots, \numprio-1 \}$
and $x=\rank_i$.
In total we have $O(\numprio)$ many $\cup$ and $O(n)$ many $\setminus$ operations. 

Notice that the update operation for a set $S_\rank$, as described above, 
also updates all sets $S_{\rank'}$ for $\rank' < \rank$. 
Thus, when using the more succinct representation via the sets $C^i_x$ 
and executing \ref{alg:SymbolicParityDominion} literally, 
the computation of the maximal rank $\rank'$ s.t.\ $S_{\rank'} \supseteq S_\rank$ 
would fail because of the earlier update of $S_\rank$. 
Hence, we have to postpone the update of $S_\rank$ till 
the end of the iteration
and adjust the computation of $\rank'$ as follows.
We do not update 
the set $S_\rank'$, and first compute the final value for $\rank'$ by 
decrementing $\rank'$ until $S_{\rank'} \supseteq  S_\rank$
and then
update $S_{\rank}$ to $S^{\text{new}}_\rank$
and thus implicitly also update the 
sets $S_{\tilde{\rank}}$ to $S_{\tilde{\rank}} \cup S_\rank$ for $\rank' <  \tilde{\rank} <\rank$. This gives the following lemma.

\begin{lemma}
In each iteration of \ref{alg:SymbolicParityDominion} only 
$O(n)$ symbolic set operations are needed to update the sets $C^i_x$,
and no symbolic one-step operation is needed.
\end{lemma}

\paragraph{Number of Set Operations.}
To sum up, when introducing the succinct representation of the sets $S_\rank$,
we only need additional $\cup$, $\cap$, and $\setminus$ operations, while 
the number of $\cpre{\pl}$ operations is unchanged.
We show in Appendix~\ref{app:pargen}
that whenever the algorithm computes or updates a set $S_\rank$,
then we can charge a $\cpre{\pl}$ operation for it, and 
each $\cpre{\pl}$ operation is only charged for a constant number of set computations and updates.
Hence, as both computing a set~$S_\rank$ and updating the sets $C^i_x$
can be done with $O(n)$ set
operations, the number of the additional set operations in \ref{alg:SymbolicParityDominion} is in $O(n \cdot \numcpre)$,
for $\numcpre$ being the number of $\cpre{\pl}$ operations in the algorithm. 

\paragraph*{Putting Things Together.}
We presented a set-based symbolic implementation of the progress measure 
that uses $O(n)$ sets, $O(\numprio \cdot n \cdot \lvert M^\infty_\domsize \rvert)$ symbolic one-step
operations and at most a factor of $n$ more symbolic set operations.
We apply the following bound
on the size of $M^\infty_\domsize$:
$$
  |M^\infty_\domsize| \leq \binom {\domsize + \lfloor \numprio / 2 \rfloor} {\domsize} + 1\,.
$$
The following lemma summarizes the result for computing dominions.

\begin{keylemma}\label{lem:small_dominions_space_efficient}
For a parity game with $n$ vertices and $\numprio$ priorities and 
an integer $\domsize \in [1, n-1]$, Algorithm~\ref{alg:SymbolicParityDominion} 
computes a player-$\pe$ dominion 
that contains all $\pe$-dominions with at most $\domsize+1$ vertices
and can be implemented with $O\left( \numprio 
\cdot n \cdot \binom {\domsize + \lfloor \numprio/2 \rfloor} {\domsize} \right)$
symbolic one-step operations, $O\left( \numprio \cdot n^2 \cdot \binom {\domsize +
\lfloor \numprio/2 \rfloor} {\domsize} \right)$ symbolic set operations,
and $O(n)$ many sets.
\end{keylemma}

To solve parity games directly with Algorithm~\ref{alg:SymbolicParityDominion},
we use the co-domain $M^\infty_\game$ instead of $M^\infty_\domsize$.
Recall
that we have
$\lvert M^\infty_\game \rvert \in O\big(\big(\frac{n}{\lfloor \numprio/2\rfloor}\big)^{\lfloor \numprio/2\rfloor}\big)$~\cite{Jurdzinski00}.

\begin{theorem}\label{thm:pmalg_space_efficient}
Let $\xi(n,\numprio)=\big(\frac{n}{\lfloor \numprio/2\rfloor}\big)^{\lfloor \numprio/2\rfloor}$.
Algorithm~\ref{alg:SymbolicParityDominion}
computes the winning sets of parity games
and can be implemented with $O\big(\numprio \cdot n \cdot \xi(n,\numprio) \big)$
symbolic one-step operations, $O\big(\numprio \cdot n^2 \cdot \xi(n,\numprio)\big)$
symbolic set operations, and $O(n)$ many sets. 
\end{theorem}

See Appendix~\ref{app:strategy_construction} for how to construct
winning strategies within the same bounds.

\section{Extensions and Conclusion}\label{sec:bigstep}
\subsection{Big-Step Algorithm.}
We presented a set-based symbolic algorithm for computing a progress measure 
that solves parity games. 
Since the progress measure algorithm can also 
compute dominions of bounded size, it can be combined with the big step approach of 
\cite{Schewe17} to improve the number of symbolic steps as stated in the 
following theorem.
All details are given in the appendix.
\begin{theorem}\label{thm:bigstep}
  Let $\gamma(\numprio) = \numprio/3 + 1/2 - 
  4/(\numprio^2 - 1)$ for odd $\numprio$ and $\gamma(\numprio) = 
  \numprio/3 + 1/2 - 1/(3 \numprio) - 4/\numprio^2$ for even $\numprio$.
  There is a symbolic Big Step Algorithm that
  computes the winning sets for parity games
  and with the minimum of $O(n \cdot (\kappa \cdot
  n/c)^{\gamma(\numprio)})$, for some constant $\kappa$,
  and $n^{O(\sqrt{n})}$ symbolic one-step operations and stores only $O(n)$ many sets.
\end{theorem}

\subsection{Concluding Remarks.}
In this work we presented improved set-based symbolic algorithms for 
parity games, and equivalently modal $\mu$-calculus model checking. 
Our main contribution improves the symbolic algorithmic complexity of one of the most
fundamental problems in the analysis of program logics, with numerous applications 
in program analysis and reactive synthesis. 
There are several practical approaches to solve parity games, 
such as,
\cite{deAlfaroF07,FriedmannL09,HuthKP11,HoffmannL13,BenerecettiDM16} and \cite{Vester16}. 
A practical direction of future work would be to explore whether our algorithmic ideas 
can be complemented with engineering efforts to obtain scalable symbolic algorithms 
for reactive synthesis of systems. 
An interesting theoretical direction of future work is to obtain set-based symbolic 
algorithms for parity games with quasi-polynomial complexity.
The breakthrough result of~\cite{CaludeJKLS17} (see also~\cite{GimbertI17}) relies on alternating 
poly-logarithmic space Turing machines.. 
The follow-up papers of \cite{JurdzinskiL17} and \cite{FearnleyJSSW17} that 
slightly improve the running time and reduce the 
space complexity from quasi-polynomial to quasi-linear rely on succinct notions of progress measures.
All these algorithms are non-symbolic, and symbolic versions of these algorithms 
are an open question, in particular encoding the novel succinct progress measures 
in the symbolic setting when storing at most a linear number of sets.

\subparagraph*{Acknowledgements.}
All authors are partially supported by the Vienna
Science and Technology Fund (WWTF) through project ICT15-003.
K. C. is partially supported by the Austrian Science Fund (FWF)
NFN Grant No S11407-N23 (RiSE/SHiNE) and an ERC Start grant
(279307: Graph Games). For W.~D., M.~H., and V.~L. the research
leading to these results has received funding from the European
Research Council under the European Union’s Seventh Framework
Programme (FP/2007-2013) / ERC Grant Agreement no. 340506.
V.L.~is partially supported by the ISF grant \#1278/16.

\pagebreak
\appendix

\section{Illustration: Parity Games with \texorpdfstring{$5$}{5} Priorities}\label{sec:par5}
\begin{figure}[b]
 \centering
 \begin{tikzpicture}
\matrix[column sep=12mm, row sep=7mm]{
	\node[p1,label=60:$1$] (a) {$a$};
	& \node[p1,label=60:$1$] (c) {$c$};
	& \node[p2,label=60:$3$] (e) {$e$};
	& \node[p1,label=60:$2$] (g) {$g$};\\
	\node[p2,label=-60:$0$] (b) {$b$};
	& \node[p1,label=-60:$0$] (d) {$d$};
	& \node[p2,label=-60:$4$] (f) {$f$};
	& \node[p1,label=-60:$1$] (h) {$h$};\\
};
\path[arrow, bend right] 
	(a) edge (b)
	(b) edge (a);
\path[arrow] 
	(c) edge (b)
	(b) edge (d)
	(d) edge (f)
	(f) edge (g)
	(g) edge (e)
	(e) edge (d)
	(c) edge (d)
	(h) edge (c)
	(h) edge (g)
	;
\end{tikzpicture}
\caption{A parity game with 5 priorities.
Circles denote player~$\pe$ vertices, squares denote player~$\po$ vertices.
The numeric label of a vertex gives its priority,
e.g., $a$ is an $\pe$-vertex with priority~$1$.}
\label{fig:example}
\end{figure}
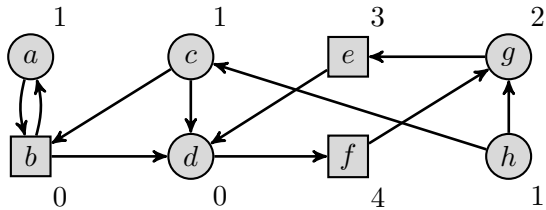
We informally introduce our symbolic algorithm to compute 
the progress measure, using parity games with $5$ priorities as an important 
special case. The pseudo-code is given in Algorithm~\ref{alg:pmfive}.
\begin{algorithm}
	\footnotesize
	\SetAlgoRefName{SymbolicProgressMeasureParity(5)}
	\caption{}
	\label{alg:pmfive}
	\SetKwInOut{Input}{Input}
	\SetKwInOut{Output}{Output}
	\SetKw{break}{break}
	\BlankLine
	\Input{%
	  \emph{parity game} $\pgame = (\game, \prio)$, with \\
	  \emph{game graph} $\game = ((V, E),(\ve, \vo))$ and \\
	  \emph{priority function} $\prio: V \rightarrow \{0,1,2,3,4\}$
	}
	\Output
	{
	  winning set of player~$\pe$ 
	}
	\BlankLine
	$S_{(0,0)} \gets V$\;
	${S_\rank} \gets {\emptyset}$ for $(0,0) < \rank \leq (n_1,n_3)$ and $\rank = \top$\;
	$\rank \gets (1,0)$\;
	\While
	  {
	    $true$
	  }
	  {
	    \If{$\rank = (i,j), i > 0$}
	    {
		$S_{(i,j)} \gets S_{(i,j)} \cup 
				    \left(
					\cpre{\po}(S_{i-1,j}) \cap P_1 
				    \right)$\label{p5l:c1add1}\;
		\Repeat 
		{
		  a fixed-point for $S_{(i,j)}$ is reached
		}
		{
		  $S_{(i,j)}  \gets S_{(i,j)}  \cup \left( \cpre{\po}(S_{i,j}) \setminus \bigcup_{i=2,3,4} P_{i} \right)$\label{p5l:c1add2}
		}
	    }
	    \ElseIf{$\rank = (0,j)$}
	    {
	        $S_{(0,j)} \gets S_{(0,j)} \cup 
				    \left(\cpre{\po}(S_{n_1,j-1}) \cap P_1 \right) \cup 
				    \left(\cpre{\po}(S_{0,j-1}) \cap P_3 \right)$\label{p5l:c2add1}\;
		\Repeat %
		{
		  a fixed-point for $S_{(0,j)}$ is reached
		}
		{
		  $S_{(0,j)}  \gets S_{(0,j)} \cup
				       \left( \cpre{\po}(S_{0,j}) \setminus P_{4} \right)$\label{p5l:c2add2}
		}
	    }
	    \ElseIf{$\rank = \top$}
	    {
		$S_{\top} \gets S_{\top} \cup
				    \left(\cpre{\po}(S_{n_1,n_3}) \cap P_1 \right) \cup 
				    \left(\cpre{\po}(S_{0,n_3}) \cap P_3 \right)$\label{p5l:c3add1}\;
		\Repeat
		{
		  a fixed-point for $S_{\top}$ is reached
		}
		{
		  $S_{\top}  \gets S_{\top}  \cup \cpre{\po}(S_\top)$\label{p5l:c3add2}		
		}
	    }
	    $\rank'\gets \decr(\rank)$\;
	    \If{$S_{\rank'} \supseteq  S_\rank$ and $\rank < \top$}
	    {
	      $\rank \gets \incr(\rank)$\label{p5l:rankinc}
	    }
	    \ElseIf{$S_{\rank'} \supseteq  S_\rank$ and $\rank = \top$}
	    {
	      \break\label{p5l:terminates}
	    }
	    \Else(\tcc*[f]{ensure $S_{\rank'} \supseteq  S_\rank$ for all $\rank' < \rank$})
	    {
	      \Repeat{$S_{\rank'} \supseteq  S_\rank$}
	      {\label{p5l:rankloop}
		  $S_{\rank '} \gets S_{\rank '} \cup S_\rank$\label{p5l:addrankinc}\;
		  $\rank'\gets \decr(\rank');$
	      }
	      $\rank \gets \incr(\rank')$\tcc*[r]{$S_\rank$ is modified set with smallest rank}\label{p5l:rankafterloop}
	    }
	  }
	 \Return $V \setminus S_\top$
\end{algorithm}
In Example~\ref{example:parity5} we use the parity game in 
Figure~\ref{fig:example} as an example to illustrate Algorithm~\ref{alg:pmfive}.

\paragraph{Intuition for ranks $\rank$ and sets $S_{\rank}$.}
We first provide some intuition for the ranks~$\rank$ 
and the sets $S_\rank$ for parity games with~$5$ priorities. 
The rank~$\rank$ has an index for each of the priorities~$\set{0, 1, 2, 3, 4}$ but
may contain non-zero entries only for the odd priorities~$\set{1, 3}$. Thus, for 
this section, we denote a rank vector as a vector with two elements, where the
first element corresponds to element~$1$ of $\rank$
and the second element corresponds to element~$3$ of~$\rank$.
The lowest rank is $(0, 0)$, followed by $(1, 0)$, the highest rank 
is $\top$, preceded by $(n_1, n_3)$.
Throughout the algorithm, whenever a vertex is contained in the set~$S_\rank$, 
then its rank in the progress measure is at least~$\rank$. 
The sets $S_{\rank'}$ are defined to contain vertices with rank 
\emph{at least}~$\rank'$ instead of \emph{exactly} $\rank'$ such that for each 
vertex~$v$ and rank~$\rank$ we only have to consider one set $S_{\rank'}$ in 
order to decide whether the rank of~$v$ can be increased to $\rank$. 
The sets $S_\rank$ implicitly assign each vertex~$v$ a rank, namely the maximum 
rank $\rank$ such that $v \in S_\rank$. 
Each vertex can only be assigned rank vectors for which the rank vector is zero 
at all indices corresponding to priorities lower than its own priority. 
For example, a vertex with priority~$4$ can only be assigned rank $(0, 0)$ or 
rank $\top$. Thus for such a vertex its rank can ``jump'' from $(0, 0)$ to $\top$
(because of one of its successors being added to $S_\top$), 
and then this information has to be propagated to its predecessors. 
The algorithm achieves this efficiently by adding, e.g., vertices that are 
added to $S_\top$ to all sets representing lower ranks as well.
We next describe the intuition for what
it means when a vertex~$v$ is assigned a specific rank~$\rank$.

Intuitively, when a vertex~$v$ is assigned a rank $(i,0)$ for 
$i \in N_1$, i.e., the set with highest rank it is contained in is~$S_{(i,0)}$, 
then, in plays starting from~$v$, player~$\po$ can force a play from $v$ to 
visit $i$~vertices of priority~$1$ before a vertex 
of priority~$2$ or higher is reached; for $i > 0$ this implies that 
the priority of $v$ is either $0$ or~$1$. 

Let vertex~$v$ be assigned rank $(i,j)$ for $i \in N_1$ and $j \in N_3$. 
The interpretation of the value 
of $i$ is the same as in the case $(i,0)$, where the value of $i$ is taken 
modulo $n_1 + 1$. Each contribution of 1 to the value of $j$ corresponds
either (a) to the number of times a priority-$3$ vertex is 
visited before a priority-4 vertex is reached or (b) to priority-$1$ vertices
being reached $n_1+ 1$ times before a vertex with priority at least~$2$ is reached.
Note that both cases can also happen for a vertex with priority 2 or 3;
for $\prio(v)= 4$ the only set~$S_{(i,j)}$ the vertex~$v$
can belong to is $S_{(0,0)}$, as the rank of a priority-4 vertex
 can only be $(0,0)$ or~$\top$. 
 
Recall that $\top$ are the vertices where player $\pe$ has no winning strategy.
 There are three ways a vertex $v$ can be ranked $\top$:
 (i) $v$ has priority-1 and the best successor is ranked $(n_1, n_3)$,
 (ii) $v$ has priority-3 and the best successor is ranked at least $(0, n_3)$,  and
 (iii) the best successor is ranked $\top$.
 The cases (i) and (ii) correspond to the cases where player $\pe$ has to visit at 
 least $n_1+1$ or $n_3+1$ many vertices of the respective 
 odd priority before reaching a higher even priority.
 Note that this means that player~$\po$ can force plays to reach a cycle where the highest priority is odd.
 In case (iii) player $\po$ can force plays to visit a vertex from which she can reach a cycle with highest odd priority.

\paragraph{Symbolic algorithm.} 
In our symbolic algorithm~\ref{alg:pmfive} we use the sets~$S_\rank$ to represent 
the numerical values of the progress measure, and, 
to utilize the power of symbolic operations, 
we compute all vertices whose rank can be increased to a certain value~$\rank$
in each iteration of the algorithm.
The latter is in contrast to the explicit progress measure algorithm~\cite{Jurdzinski00},
where vertices are considered one by one and the rank is increased to the maximal possible value. 

We next describe Algorithm~\ref{alg:pmfive} and then give some intuition for its 
correctness and the number of its symbolic operations.
Recall that the sets~$S_\rank$ define a ranking function
$\rho_{\{S_\rank\}_\rank}(v) = \max \{r \in M^\infty_\game \mid v \in S_\rank \}$
that assigns a rank to each vertex and that by the definition of $\incr$
 a vertex~$v$ with priority $\prio(v)$ is only assigned
ranks $\rank$ with $\rank = \proj{\rank}_{\prio(v)}$.

\paragraph{Initialization.}
To find the least simultaneous fixed point of the lift-operators, all ranks are initialized 
with the all zero vector, i.e., all vertices are added to $S_{(0,0)}$, while 
all other sets $S_\rank$ are empty. The variable~$\rank$ is initialized to~$(1,0)$.

\paragraph{The value of~$\rank$.}
In each iteration of the while loop the set~$S_\rank$ for the rank that is stored in 
the variable~$\rank$ at the beginning of the iteration 
is updated (more details below). 
By the definition of the sets $S_{\rank'}$,
we need to maintain $S_{\rank'} \supseteq S_\rank$ for $\rank' < \rank$, i.e.,
every vertex that is newly added to $S_\rank$ but not 
yet contained in $S_{\rank'}$ is added to $S_{\rank'}$ 
(line~\ref{p5l:addrankinc}). For the vertices newly added to $S_\rank$ we 
say that their rank is increased. When the rank of a vertex is increased, this 
might influence the value of $\lift(\rho_{\{S_\rank\}_\rank}, v)(v)$ for 
its predecessors~$v$. Since we want to obtain a fixed point of 
$\lift(\rho_{\{S_\rank\}_\rank}, v)(v)$ for all $v \in V$, we have to reconsider 
the predecessors of a vertex whenever the rank of the vertex is increased.
This is achieved by updating the variable~$\rank$ to the lowest~$\rank'$
for which a new vertex is added to $S_{\rank'}$ in this iteration (lines~\ref{p5l:rankloop}--\ref{p5l:rankafterloop}).
If $\rank = \rank'$, then for $\rank < \top$ the value of $\rank$ is increased
to the next highest rank in the ordering (line~\ref{p5l:rankinc}) and for $\rank = \top$ the algorithm terminates (line~\ref{p5l:terminates}).

\paragraph{Update of set~$S_\rank$.}
To reach a simultaneous fixed point of the lift-operators, the rank of a vertex~$v$
has to be increased to $\lift(\rho_{\{S_\rank\}_\rank}, v)(v)$
whenever the value of $\lift(\rho_{\{S_\rank\}_\rank}, v)(v)$ is strictly higher 
than $\rho_{\{S_\rank\}_\rank}(v)$ for the current ranking function
$\rho_{\{S_\rank\}_\rank}$. Recall that this is the case only if the value of 
$\best(\rho, v)$ is increased, which implies that the rank assigned to 
at least one successor of $v$ was increased (or a rank~$\rank$ is considered for 
the first time).
For the update of the set~$S_\rank$ in an iteration of the while-loop
(where $\rank$ is the value of the variable at the beginning of the while-loop),
the algorithm adds to $S_\rank$ all vertices that satisfy the following three conditions:
(i)~$\rho_{\{S_\rank\}_\rank}(v) < \rank$,  
(ii) $\lift(\rho_{\{S_\rank\}_\rank}, v)(v) \ge \rank$, and (iii) 
$\rank = \proj{\rank}_{\prio(v)}$.
Note that the algorithm maintains the invariant that vertices 
with $\rho_{\{S_\rank\}_\rank} \ge \rank$ are already contained in $S_\rank$.
We distinguish between $\rank = (i, j)$ with $i > 0$, $\rank = (0, j)$,
and $\rank = \top$. 

(1)
For $\rank = (i, j)$ with $i > 0$ recall that 
only vertices~$v$ with priority 0 or 1 can be assigned the rank~$\rank$. 

\emph{Case~\upbr{a}}: Assume first that $\prio(v) = 1$, i.e., 
$\lift(\rho_{\{S_\rank\}_\rank}, v)(v) = \min\{y \in M^\infty_\game \mid y >_1 
\best(\rho_{\{S_\rank\}_\rank},v)\}$.
Then $\lift(\rho_{\{S_\rank\}_\rank}, v)(v)\ge \rank$
if $\best(\rho_{\{S_\rank\}_\rank},v) \ge (i-1, j)$, i.e., if
(i) $v \in \ve$ and all successors~$w$ of~$v$ have $\rho_{\{S_\rank\}_\rank}(w)
\ge (i-1, j)$ or (ii) $v \in \vo$ and one successor~$w$ of $v$ has 
$\rho_{\{S_\rank\}_\rank}(w) \ge (i-1, j)$.
That is, $\lift(\rho_{\{S_\rank\}_\rank}, v)(v)\ge \rank$ only if $v \in 
\cpre{\po}(S_{(i-1, j)})$. In this case the vertex $v$ is added to $S_\rank$ in 
line~\ref{p5l:c1add1}.

\emph{Case~\upbr{b}}: Assume now that $\prio(v) = 0$, i.e., $\lift(\rho_{\{S_\rank\}_\rank}, v)(v) = \proj{\best(\rho_{\{S_\rank\}_\rank},v)}_0
= \best(\rho_{\{S_\rank\}_\rank},v)$.
Then $\lift(\rho_{\{S_\rank\}_\rank}, v)(v)\ge \rank$
if $\best(\rho_{\{S_\rank\}_\rank},v) \ge (i, j)$, i.e., if 
(i) $v \in \ve$ and all successors~$w$ of~$v$ have $\rho_{\{S_\rank\}_\rank}(w)
\ge (i, j)$ or (ii) $v \in \vo$ and one successor~$w$ of $v$ has 
$\rho_{\{S_\rank\}_\rank}(w) \ge (i, j)$.
That is, $\lift(\rho_{\{S_\rank\}_\rank}, v)(v)\ge \rank$ only if $v \in 
\cpre{\po}(S_{(i, j)})$. In this case the vertex $v$ is added to $S_\rank$
in some iteration of the repeat-until loop in line~\ref{p5l:c1add2}.

(2)
The difference for $\rank = (0, j)$ is first that also vertices 
with priorities 2 or 3 are candidates for the assignment of rank~$\rank$ and second
that for a vertex~$v$ with odd priority we now have the following possibilities:
either (i) $\prio(v) = 1$ and we consider neighbors~$w$ with 
$\rho_{\{S_\rank\}_\rank}(w) \ge (n_1, j-1)$ or (ii) $\prio(v) = 3$ and 
we consider neighbors~$w$ with $\rho_{\{S_\rank\}_\rank}(w) \ge (0, j-1)$.

(3)
The case $\rank = \top$ corresponds to the case $\rank = (0, j)$ with 
$j = n_3 + 1$ with the difference that also vertices with priority~4 can be
included in $S_\top$ in the case $\rho_{\{S_\rank\}_\rank}(w) = \top$.

\begin{example}\label{example:parity5}
In this example we apply Algorithm~\ref{alg:pmfive} to the parity game in 
Figure~\ref{fig:example}. We have $n_1=3$ and $n_3=1$ and thus we have to consider ranks in the co-domain 
$M^\infty_G=\{(0,0),(1,0),(2,0),(3,0),(0,1),$ $(1,1),(2,1),(3,1),\top\}$. \smallskip

\noindent The algorithm initializes the set $S_{(0,0)}$ to $\{a,b,c,d,e,f,g,h\}$
and $\rank$ to $(1,0)$. All the other sets $S_\rank$ are initialized as the empty set.
It then proceeds as follows:
\begin{enumerate}
\item In the first iteration of the while-loop it processes $\rank=(1,0)$.
In line~\ref{p5l:c1add1} the vertices $a,c$ and $h$ are added to $S_{(1,0)}$
and then in line~\ref{p5l:c1add2} also $b$ is added.
That is, we obtain $S_{(1,0)}=\{a,b,c,h\}$ and as $S_{(1,0)} \subseteq  S_{(0,0)}$,
the rank $\rank$ is increased to $(2,0)$.

\item In the second iteration it processes $\rank=(2,0)$ and
the vertex $a$ is added to $S_{(2,0)}$ in line~\ref{p5l:c1add1} 
and the vertex $b$ is added to $S_{(2,0)}$ in line~\ref{p5l:c1add2},
i.e., $S_{(2,0)}=\{a,b\}$, and $\rank$ is set to $(3,0)$.

\item When processing $\rank=(3,0)$ the set $S_{(3,0)}$ is updated to 
$\{a,b\}$ and $\rank$ is increased to $(0,1)$.

\item Now the algorithm processes the rank $(0,1)$ the first time.
     The vertex $a$ is added to $S_{(0,1)}$ because of the 
     $\left(\cpre{\po}(S_{n_1,j-1}) \cap P_1 \right)$ part in line~\ref{p5l:c2add1}
     and the vertex $e$ is added because of the $\left(\cpre{\po}(S_{0,j-1}) \cap P_3 \right)$
     part in line~\ref{p5l:c2add1}. Finally, also $b$ and $g$ are added in line~\ref{p5l:c2add2}.
     That is, we have  $S_{(0,1)}=\{a,b,e,g\}$.
     Now as $S_{(3,0)}\not\subseteq S_{(0,1)}$, $S_{(2,0)}\not\subseteq S_{(0,1)}$ and $S_{(1,0)}\not\subseteq S_{(0,1)}$
     we have to decrease $\rank$ to $(1,0)$, and also to modify the other 
     sets with smaller rank as follows: $S_{(1,0)}=\{a,b,c,e,g,h\}$; and
     $S_{(2,0)}=S_{(3,0)}=\{a,b,e,g\}$.
\item The algorithm considers $\rank=(1,0)$ again, makes no changes to $S_{(1,0)}$,
      and sets $\rank$ to $(2,0)$.
\item Now considering $\rank=(2,0)$, the vertex $h$ is added to the set $S_{(2,0)}$
      in line~\ref{p5l:c1add1}, i.e.,  $S_{(2,0)}=\{a,b,e,g,h\}$,
      and, as $h$ is already contained in  $S_{(1,0)}$, $\rank$ is increased to $(3,0)$.
\item The set $S_{(3,0)}$ is not changed and $\rank$ is increased to
      $(0,1)$.
\item The set $S_{(0,1)}$ is not changed and $\rank$ is increased to
      $(1,1)$.
\item The vertex $a$ is added to $S_{(1,1)}$ in line~\ref{p5l:c1add1} 
      and the vertex $b$ is added to $S_{(1,1)}$ in line~\ref{p5l:c1add2},
      i.e., $S_{(1,1)}=\{a,b\}$, and $\rank$ is increased to $(2,1)$.
\item The vertices $a,b$ are added to $S_{(2,1)}$, i.e., $S_{(2,1)}=\{a,b\}$,
      and $\rank$ is increased to $(3,1)$.      
\item The vertices $a,b$ are added to $S_{(3,1)}$, i.e., $S_{(3,1)}=\{a,b\}$,
      and $\rank$ is increased to $\top$.      
\item The vertex $a$ is added to $S_{\top}$ in line~\ref{p5l:c3add1} 
      and $b$ is added to $S_{\top}$ in line~\ref{p5l:c3add2},
      i.e., $S_\top=\{a,b\}$.
      
      Now as $S_{(3,1)} \subseteq S_\top$, the algorithm terminates.
\end{enumerate}\smallskip

Finally we have that 
$S_{(0,0)} = \{a,b,c,d,e,f,g,h\}$,
$S_{(1,0)} = \{a,b,c,e,g,h\}$,
$S_{(2,0)} = \{a,b,e,g,h\}$,
$S_{(3,0)} = S_{(0,1)} = \{a,b,e,g\}$,
and 
$S_{(1,1)} = S_{(2,1)} = S_{(3,1)} = S_{\top} = \{a,b\}$.
That is, the algorithm returns $\{c,d,e,f,g,h\}$ as the winning set of player~$\pe$.
The final sets of the algorithm corresponds to the progress measure $\rho$ with
$\rho(f)=\rho(d)=(0,0)$, 
$\rho(c)=(1,0)$,
$\rho(h)=(2,0)$,
$\rho(e)=\rho(g)=(0,1)$, and
$\rho(a)=\rho(b)=\top$.
\qee 
\end{example}

\paragraph{Sketch of bound on number of symbolic operations.}
Observe that each rank~$\rank$ is considered in at least one iteration of the 
while-loop but is only reconsidered in a later iteration 
if at least one vertex was added to the set $S_\rank$ since the last time $\rank$
was considered.
This implies by $\lvert S_\rank \rvert \le n$
that Algorithm~\ref{alg:pmfive} can be implemented with 
$O(n \cdot |M^\infty_\game|) = O(n^3)$ symbolic operations.

\paragraph{Sketch of correctness.} 
Let $\{S_\rank\}_\rank$ be the sets in the algorithm at termination.
The algorithm returns the set of vertices that are not contained in $S_\top$,
i.e., the vertices to which $\rho_{\{S_\rank\}_\rank}$  assigns a rank $< \top$. 
If $\rho_{\{S_\rank\}_\rank}$ is equal to the progress measure of the parity game, 
then by Lemma~\ref{lem:jur:pmcorrect} the returned set $V \setminus S_\top$
is equal to the winning set of player~$\pe$. 
Let $\tilde{\rho}$ denote the progress measure.
It remains to show that $\rho_{\{S_\rank\}_\rank}(v) = \tilde{\rho}(v)$ for 
all $v \in V$ when the algorithm terminates. To this end we show
(1) that the algorithm only adds a vertex to a set~$S_\rank$ 
when the progress measure~$\tilde{\rho}$ is at least~$\rank$, i.e.,
throughout the algorithm the ranking function $\rho_{\{S_\rank\}_\rank}$
is a lower bound on~$\tilde{\rho}$ and (2) by the update of the variable 
$\rank$ we have before and after each iteration of the while-loop 
for all vertices~$v$ that either 
$\lift(\rho_{\{S_\rank\}_\rank},v)(v) = \rho_{\{S_\rank\}_\rank}(v)$
or $\rho_{\{S_\rank\}_\rank}(v) \geq \rank$ and in the final iteration
with $\rank = \top$ we have $\lift(\rho_{\{S_\rank\}_\rank},v)(v) = 
\rho_{\{S_\rank\}_\rank}(v)$ for all $v \in V$.
This implies that when the algorithm terminates 
the ranking function $\rho_{\{S_\rank\}_\rank}$ is a simultaneous fixed point of 
the lift-operators. Together these two properties imply that the algorithm 
computes the progress measure of the parity game.

\section{Details for Section~\ref{sec:pmalg}}\label{app:pargen}
\paragraph{Correctness.} 
The correctness of Algorithm~\ref{alg:SymbolicParityDominion},
stated in the following lemma,
follows from combining Lemma~\ref{lem:ifinv} with
Lemmata~\ref{lem:dom_invariant2}--\ref{lem:dom_invariant3}, which we prove below.
\begin{lemma}[Correctness]\label{lem:SymbolicParityDominion_correctness}
  Algorithm~\ref{alg:SymbolicParityDominion} computes the progress measure 
for a given parity game \upbr{with $n$ vertices} 
and a given set of possible ranks $M^\infty_\domsize$
\upbr{for some integer~$\domsize \in [1, n-1]$}.
\end{lemma}
\begin{lemma}\label{lem:ifinv}
	Assuming that Invariant~\ref{invariant} holds, the ranking function
$\rho_{\{S_\rank\}_\rank}$ induced by the family of sets $\{S_\rank\}_\rank$ at termination 
of Algorithm~\ref{alg:SymbolicParityDominion} is equal to 
the progress measure for the given parity game and the co-domain
$M^\infty_\domsize$.  
\end{lemma}

\begin{proof}
Recall that the progress measure is the 
least simultaneous fixed point of all $\lift(., v)$-operators
for the given parity game (where $\incr$, $\decr$, and the ordering of ranks
are w.r.t.\ the given co-domain) and let the progress measure be denoted 
by~$\tilde{\rho}$. Let $\{S_\rank\}_\rank$ be the sets in the algorithm 
at termination.
For all $v \in V$ the ranking function $\rho_{\{S_\rank\}_\rank}(v)$ is defined as 
$ \max \{\rank \in M^\infty_h \mid v \in S_\rank \}$. 
By Invariant~\ref{invariant}(\ref{inv:atleast})
we have $\rho_{\{S_\rank\}_\rank}(v) \le \tilde{\rho}(v)$ for all $v \in V$.

When the algorithm terminates, with $\rank = \top$, we have 
by Invariant~\ref{invariant}(\ref{inv:closure}) 
$\lift(\rho_{\{S_\rank\}_\rank},v)(v)  = \rho_{\{S_\rank\}_\rank}(v)$
for each vertex~$v$ and thus $\rho_{\{S_\rank\}_\rank}$ is a simultaneous
fixed point of the $\lift(.,v)$-operators.
Now, as $\tilde{\rho}$ is the least simultaneous fixed point of all $\lift(., v)$-operators,
we obtain $\rho_{\{S_\rank\}_\rank}(v) \ge \tilde{\rho}(v)$ for all $v \in V$.
Hence we have $\rho_{\{S_\rank\}_\rank}(v) = \tilde{\rho}(v)$ for all $v \in V$.
\end{proof}

\begin{lemma}\label{lem:dom_invariant2}
Before and after each iteration of the while-loop in 
 Algorithm~\ref{alg:SymbolicParityDominion}
 we have $S_{\rank_1} \supseteq S_{\rank_2}$ for all $\rank_1 \leq \rank_2$
 with $\rank_1, \rank_2 \in M^\infty_\domsize$,
 i.e., Invariant~\ref{invariant}\upbr{\ref{inv:monotonicity}} holds.
\end{lemma}

\begin{proof}
The proof is by induction over the iterations of the while-loop.
The claim is 
satisfied when we first enter the while-loop and only $S_{\bar{0}}$ is non-empty.   
It remains to show that when the claim is valid at the beginning of a iteration
then the claim also hold afterwards.
By the induction hypothesis, the sets $S_{\rank'}$ for $\rank' < \rank$ 
are monotonically decreasing. Thus it is sufficient to find the lowest 
rank~$\rank^*$ such that for all $\rank^* \le \rank' < \rank$ 
we have $S_\rank \not\subseteq S_{\rank'}$ and 
add the vertices newly added to $S_\rank$ to the sets $S_{\rank'}$ 
with $\rank^* \le \rank' < \rank$, which is done in
lines~\ref{alg:symbParityc_line16}--\ref{alg:symbParityc_line25} of the while-loop.
\end{proof}

\begin{lemma}\label{lem:dom_invariant1}
 Let $\tilde{\rho}$ be the progress measure of the given parity game
 and let $\rho_{\{S_{\rank}\}_{\rank}}(v)= 
\max \{\rank \in M^\infty_h \mid v \in S_{\rank} \}$
be the ranking function with respect to the family of sets $\{S_{\rank}\}_{\rank}$
that is maintained by the algorithm. Throughout Algorithm~\ref{alg:SymbolicParityDominion} we have
$\tilde{\rho}(v) \geq \rho_{\{S_{\rank}\}_{\rank}}(v)$ for all $v \in V$, i.e., 
Invariant~\ref{invariant}\upbr{\ref{inv:atleast}} holds.
\end{lemma}

\begin{proof}
We show the lemma by induction over the iterations of the while-loop.
Before the first iteration of the while-loop only $S_{\bar{0}}$ is non-empty,
thus the claim holds by $\tilde{\rho} \ge \bar{0}$.

Assume we have $\rho_{\{S_{\rank}\}_{\rank}}(v) \le \tilde{\rho}(v)$ 
for all $v \in V$ before an iteration of the while-loop. We show
that $\rho_{\{S_{\rank}\}_{\rank}}(v) \le \tilde{\rho}(v)$ 
also holds during and after the iteration of the while-loop.
As the update of $S_{\rank'}$ in line~\ref{alg:symbParityc_line22}
does not change $\rho_{\{S_{\rank}\}_{\rank}}$, we only have to show 
that the invariant is maintained by the update of $S_\rank$ in 
lines~\ref{alg:startupdateS}--\ref{alg:endupdatetop}.
Further $\rho_{\{S_{\rank}\}_{\rank}}(v)$ only changes for vertices
newly added to $S_\rank$, thus we only have to take these vertices into account.

  Let $\ell$ be the maximal index such that $\rank=\proj{\rank}_\ell$
  or the highest odd priority if $\rank = \top$.
  Assume $\rank < \top$, the argument for $\rank = \top$ 
  is analogous. The algorithm adds vertices to $S_\rank$
  in (1) line~\ref{alg:symbParityc_line6} and 
  (2) line~\ref{alg:symbParityc_line8}.
  In case~(1) we add the vertices $\bigcup_{1 \le \idx \leq (\ell+1)/2}
	(
	\cpre{\po}(S_{{\decr_{2\idx-1}(\rank)}}) \cap P_{2\idx-1} 
	)
	$
  to $S_\rank$.
  Let $v \in \cpre{\po}$ $(S_{{\decr_{2\idx-1}(\rank)}}) \cap P_{2\idx-1}$ 
  for some $1 \le \idx \leq (\ell+1)/2$.
  \begin{itemize}
    \item If $v \in V_\pe \cap P_{2\idx-1}$, then all successors $w$ of $v$ are 
    in $S_{{\decr_{2\idx-1}(\rank)}}$ and thus, by the induction hypothesis,
	  have $\tilde{\rho}(w) \geq \decr_{2\idx-1}(\rank)$. Now as $v \in P_{2\idx-1}$,
	  it has rank $\tilde{\rho}(v)$ at least $\incr_{2\idx-1}(\decr_{2\idx-1}(\rank))= \rank$.
    \item If $v \in V_\po \cap P_{2\idx-1}$, at least one successors $w$ of $v$ 
    is in $S_{{\decr_{2\idx-1}(\rank)}}$ and thus, by the induction hypothesis,
	  has $\tilde{\rho}(w) \geq \decr_{2\idx-1}(\rank)$. Now as $v \in P_{2\idx-1}$, 
	  it has rank $\tilde{\rho}(v)$ at least $\incr_{2\idx-1}(\decr_{2\idx-1}(\rank))= \rank$.
  \end{itemize}
  
  \noindent For case~(2) consider a vertex $v \in  \cpre{\po}(S_\rank) \setminus \bigcup_{\ell < \idx \le d} P_\idx$ added in line~\ref{alg:symbParityc_line8}.
  \begin{itemize}
    \item If $v \in V_\pe$, all successors $w$ of $v$ are in $S_\rank$ and thus, 
    by the induction hypothesis, have $\tilde{\rho}(w) \geq \rank$. Since the
    priority of $v$ is $\leq \ell$, we have $\tilde{\rho}(v) \geq \proj{\rank}_\ell =  \rank$.
    \item If $v \in V_\po$, at least one successors $w$ of $v$ is in $S_\rank$ and 
    thus, by the induction hypothesis, has $\tilde{\rho}(w) \geq \rank$. 
    Since the priority of $v$ is $\leq \ell$, 
	  we have $\tilde{\rho}(v) \geq \proj{\rank}_\ell = \rank$.\qedhere
  \end{itemize}
\end{proof}

\begin{lemma}\label{lem:dom_invariant3}
Before and after each iteration of the while loop
we have for the rank
			stored in $\rank$ and all vertices~$v$ either 
			$\lift(\rho_{\{S_\rank\}_\rank},v)(v)  \ge  \rank$ or
			$\lift(\rho_{\{S_\rank\}_\rank},v)(v) = \rho_{\{S_\rank\}_\rank}(v)$.
At line~\ref{alg:symbParityc_line15} of the algorithm we additionally have 
			$v \in S_\rank$ for all vertices~$v$ for which the value of 
			$\lift(\rho_{\{S_\rank\}_\rank},v)(v)$ is equal to~$\rank$.
Thus Invariant~\ref{invariant}\upbr{\ref{inv:closure}} holds.
\end{lemma}

\begin{proof}
  We show the claim by induction over the iterations of the while-loop.
  Before we first enter the loop, we have $\rank = \incr(\bar{0})$ and $S_{\bar{0}}=V$
  and thus the claim is satisfied.
  For the inductive step, let $\rank^\text{old}$ be the value of~$\rank$ and $\rho^\text{old}$ the ranking 
  function $\rho_{\{S_\rank\}_\rank}$ before a fixed iteration of the 
  while-loop and assume we have for all $v \in V$ either 
  $\lift(\rho^\text{old},v)(v)  \ge  \rank^\text{old}$ or
  $\lift(\rho^\text{old},v)(v) = \rho^\text{old}(v)$ before
  the iteration of the while-loop. Let $\rank^\text{new}$ be the value of $\rank$ and 
  $\rho^\text{new}$ the ranking function $\rho_{\{S_\rank\}_\rank}$ after the
  iteration. 
  We have three cases for the value of $\rank^\text{new}$:
  (1)~$\rank^\text{new} = \incr(\rank^\text{old})$ (line~\ref{alg:symbParityc_line16}),
  (2)~$\rank^\text{new} = \rank^\text{old} = \top$ (line~\ref{alg:ifterminates}), or
  (3)~$\rank^\text{new} < \rank^\text{old}$, i.e., the rank is decreased in 
    lines~\ref{alg:rankloop}--\ref{alg:rankafterloop} to maintain anti-monotonicity.

  We show in Claim~\ref{claim:27} that, in all three cases, if a set  $S_{\rank'}$, for some $\rank' < \rank^\text{old}$, 
  is not changed in the considered iteration of the while-loop then 
  for all $v \in V$ with $\lift(\rho^\text{new},v)(v) \le \rank'$ we have that 
  $\lift(\rho^\text{new},v)(v) = \rho^\text{new}(v)$. 

  Given Claim~\ref{claim:27}, we prove the first part of the invariant as follows.
  In the {case~\upbr{1}} the lowest (and only) rank for which the set is updated is $\rank^\text{old}$,
  thus it remains to show $\lift(\rho^\text{new},v)(v) = \rho^\text{new}(v)$
  for vertices with $\lift(\rho^\text{new},v)(v) = \rank^\text{old}$, which
  is done by showing the second part of the invariant, namely that
  $v \in S_{\rank^\text{old}}$ for all vertices~$v$ with
  $\lift(\rho_{\{S_\rank\}_\rank},v)(v) = \rank^\text{old}$ after the update of the 
  set $S_{\rank^\text{old}}$ in lines~\ref{alg:startupdateS}--\ref{alg:endupdatetop};
  for case~(1) we have $\rho^\text{new} = \rho_{\{S_\rank\}_\rank}$ at this point
  in the algorithm.

  In the {cases~\upbr{2} and~\upbr{3}} we have that the lowest rank for which the set is updated in 
  the iteration is equal to $\rank^\text{new}$, thus Claim~\ref{claim:27} implies 
  that the invariant $\lift(\rho^\text{new},v)(v)  \ge  \rank^\text{new}$ or
  $\lift(\rho^\text{new},v)(v) = \rho^\text{new}(v)$ holds for all 
  $v \in V$ after the while-loop.

  \begin{claim}\label{claim:27}
    Let $\rank^* \le \rank^\text{old}$ be a rank with the guarantee that no set 
    corresponding to a lower rank than $\rank^*$ is changed in this iteration of 
    the while-loop. Then we have for all $v \in V$ either 
    $\lift(\rho^\text{new},v)(v)  \ge  \rank^*$ or
    $\lift(\rho^\text{new},v)(v) = \rho^\text{new}(v)$ after
    the iteration of the while-loop.
  \end{claim}

      To prove the claim, note that since
      each set $S_\rank$ is monotonically non-decreasing over 
      the algorithm, we have $\rho^\text{new}(v) \ge \rho^\text{old}(v)$
      and $\lift(\rho^\text{new},v)(v) \ge \rho^\text{new}(v)$. 
      Assume by contradiction that there is a vertex~$v$ with 
      $\lift(\rho^\text{new},v)(v) > \rho^\text{new}(v)$ and 
      $\lift(\rho^\text{new},v)(v) <  \rank^*$.
      The latter implies $\best(\rho^\text{new},v) <  \rank^*$.
      By the induction hypothesis for $\rank^* \le \rank^\text{old}$ we have
      $\lift(\rho^\text{old},v)(v) = \rho^\text{old}(v)$.
      By $\rho^\text{new}(v) \ge \rho^\text{old}(v)$ and the definition of the 
      lift-operator this implies $\best(\rho^\text{new},v) > \best(\rho^\text{old},v)$,
      i.e., the rank assigned to at least one vertex~$w$ with $(v, w) \in E$ 
      is increased. By  
      $\best(\rho^\text{new},v) <  \rank^*$
      this implies that a set $S_{\rank'}$ with $\rank' <  \rank^*$
      is changed in this iteration, a contradiction to the definition of $\rank^*$.
      This concludes the proof of the claim.\smallskip

  It remains to show that $v \in S_{\rank^\text{old}}$ for all vertices~$v$ with
  $\lift(\rho_{\{S_\rank\}_\rank},v)(v) = \rank^\text{old}$ after the update of the 
  set $S_{\rank^\text{old}}$ in lines~\ref{alg:startupdateS}--\ref{alg:endupdatetop}.
  Towards a contradiction assume that there is a $v \not\in S_{\rank^\text{old}}$ 
  such that $\lift(\rho_{\{S_\rank\}_\rank},v)(v) = \rank^\text{old}$.
  Assume $v \in V_\pe$, the argument for $v \in V_\po$ is analogous.
  Let $\ell$ be maximal such that $\rank^\text{old}
  =\proj{\rank^\text{old}}_\ell$ for $\rank^\text{old} < \top$
  and let $\ell$ be the highest odd priority in the parity game 
  for $\rank^\text{old} = \top$. Notice that $\prio(v)$ can be at most
  $\ell$ for $\lift(\rho_{\{S_\rank\}_\rank},v)(v) = \rank^\text{old}$ to hold.
  We now distinguish two cases depending on whether $\prio(v)$ is odd or even.
  \begin{itemize}
    \item If $\prio(v)$ is odd, i.e., $\prio(v)=2\idx-1$ for some $\idx \le (\ell+1)/2$,
      then we have that all successors $w$ of $v$ have 
      $\rho_{\{S_\rank\}_\rank}(w) \geq \decr_{2\idx-1}(\rank^\text{old})$ and 
      thus, by Lemma~\ref{lem:dom_invariant2}, 
      $w \in S_{\decr_{2\idx-1}(\rank^\text{old})}$. 
      But then $v$ would have being included in $S_{\rank^\text{old}}$ in 
      line~\ref{alg:symbParityc_line6}, a contradiction.

    \item If $\prio(v)$ is even, i.e., $\prio(v)=2\idx$ for some $\idx \le \ell/2$,
      then we have that all successors $w$ of $v$ have $\rho_{\{S_\rank\}_\rank}(w) 
      \geq \rank^\text{old}$.
      Then by the definition of 
      $\rho_{\{S_\rank\}_\rank}$ it must be that $w \in S_{\rank'}$ 
      for some $\rank' \geq \rank^\text{old}$ and 
      by Lemma~\ref{lem:dom_invariant2} it must be that $w \in S_{\rank^\text{old}}$. 
      But then $v$ would have being included in $S_{\rank^\text{old}}$ in 
      line~\ref{alg:symbParityc_line8}, a contradiction.
  \end{itemize}
  Thus, after the update of the set $S_{\rank^\text{old}}$ we have that $v \in S_{\rank^\text{old}}$ for all vertices~$v$ with
  $\lift(\rho_{\{S_\rank\}_\rank},v)(v) = \rank^\text{old}$.
  Together we the above observations this proves the lemma.
\end{proof}

\subparagraph*{Number of symbolic operations.} 
We address next the number of symbolic operations of Algorithm~\ref{alg:SymbolicParityDominion} when using the sets $S_\rank$ directly.
We analyze the number of symbolic operations when using a linear number of sets below.
The main idea is that a set $S_\rank$ is only reconsidered if at least 
one new vertex was added to $S_\rank$.
\begin{lemma}
  For parity games with $n$ vertices and $\numprio$ priorities
  Algorithm~\ref{alg:SymbolicParityDominion} takes $O(\numprio \cdot n \cdot |
  M^\infty_\domsize|)$ many symbolic operations and uses  $O(|M^\infty_\domsize|)$ many sets, 
  where $\domsize$ is some integer in $[1, n-1]$.
\end{lemma}
\begin{proof}
In the algorithm we use one set $S_\rank$ for each $\rank \in M^\infty_\domsize$ and 
 thus $|M^\infty_\domsize|$ many sets.
 We first consider the number of symbolic operations needed to compute the sets $S_\rank$
 in lines~\ref{alg:startupdateS}--\ref{alg:endupdatetop}, and 
 then consider the number of symbolic operations to compute the new value of $\rank$ 
 in lines~\ref{alg:symbParityc_line15}--\ref{alg:symbParityc_line25}.
 
 \begin{enumerate}
 \item[1)] Whenever we consider a set $S_\rank$, we first initialize the set
 with $O(\numprio)$ many symbolic operations 
 (lines~\ref{alg:symbParityc_line6} \& \ref{alg:symbParityc_line11}).
 After that we do a fixed-point computation that needs symbolic operations 
 proportional to the number of added vertices.
 Now fix a set $S_\rank$ and consider all the fixed-point computations for $S_\rank$
 over the whole algorithm.
 As only $O(n)$ many vertices can be added to $S_\rank$, all
 these fixed-points can be computed in $O(n + \#r)$ symbolic operations, 
 where $\#\rank$ is the number of times the set $S_\rank$ is considered 
 (the algorithm needs a constant number of symbolic operations to realize that a fixed-point was already reached).
 Each set $S_\rank$ is considered at least once and only reconsidered when some new vertices are added to the set, i.e., 
 it is considered at most $n$ times. Thus for each set $S_\rank$ we have
 $O(\numprio \cdot n)$ many operations,
 which gives a total number of operations of 
 $O(\numprio \cdot n \cdot |M^\infty_\domsize|)$.
 
 \item[2)] Now consider the computation of the new value of $\rank$ in 
lines~\ref{alg:symbParityc_line15}--\ref{alg:symbParityc_line25}.
 Lines~\ref{alg:symbParityc_line15}--\ref{alg:terminates} take a constant number
 of operations. It remains to count the iterations of 
the repeat-until loop in 
lines~\ref{alg:symbParityc_line20}--\ref{alg:symbParityc_line25},
which we bound by the number of iterations of the while-loop as follows.
 Whenever a set $S_{\rank'}$ is considered as the left side argument in 
 line~\ref{alg:symbParityc_line24},
 then the new value for $\rank$ is less or equal to $\incr(\rank')$ and 
 thus there will be another iteration of the while-loop considering $\incr(\rank')$. 
 As there are only $O(n \cdot |M^\infty_\domsize|)$ many iterations
 of the while-loop
 over the whole algorithm, there are only $O(n \cdot |M^\infty_\domsize|)$
 many iterations of the repeat-until loop in total. In each iteration a constant
 number of operations is performed.
\end{enumerate}
 By (1) and (2) we have that Algorithm~\ref{alg:SymbolicParityDominion} takes 
 $O(\numprio \cdot n \cdot \left|M^\infty_\domsize\right|)$ many symbolic operations.
\end{proof}

\paragraph{Number of set operations in linear space algorithm.}
For the proof of Lemma~\ref{lem:small_dominions_space_efficient} and 
thus of Theorem~\ref{thm:pmalg_space_efficient} it remains to show 
that whenever the algorithm computes or updates a set $S_\rank$ using 
the succinct representation with the sets
$C^i_x$ introduced in Section~\ref{sec:space},
then we can charge a $\cpre{\pl}$ operation for it, and 
each $\cpre{\pl}$ operation is only charged for a constant number of 
set computations and updates. The argument is as follows.
(a)~Whenever the algorithm computes a 
set~$S_\rank$ in line~\ref{alg:symbParityc_line6} or~\ref{alg:symbParityc_line11},
at least one $\cpre{\pl}$ computation with this set is done. 
(b)~Now consider the computation of the new value of $\rank$.
The subset tests in lines~\ref{alg:symbParityc_line16} and~\ref{alg:ifterminates} 
are between a set that was already computed in line~\ref{alg:symbParityc_line6} or~\ref{alg:symbParityc_line11} and 
the set computed in line~\ref{alg:symbParityc_line8} or~\ref{alg:symbParityc_line13} and thus, 
if we store these sets, we do not require additional operations.
Whenever a set $S_{\rank'}$ is considered as the left side argument in 
line~\ref{alg:symbParityc_line24},
then the new value of $\rank$ is less or equal to $\incr(\rank')$ and 
thus there will be another iteration of the while-loop considering $\incr(\rank')$. 
Hence, we can charge the additional operations needed for the comparison 
to the $\cpre{\pl}$ operations of the next iteration that processes the 
rank $\incr(\rank')$.
(c)~Finally, we only need to update the sets $C^i_x$ once per iteration and in 
each iteration we perform at least one $\cpre{\pl}$ computation 
that we can charge for the update.

\section{Traps, Attractors, and Dominions}\label{sec:concepts}
For the algorithms in Appendices~\ref{sec:classic} and~\ref{app:bigstep} we use the 
following well-known notions. An example is given below. 

\paragraph{Traps.}
A set $U \subseteq V$ is a \emph{$\pl$-trap} if for all $\pl$-vertices 
$u$ in $U$ we have $\Out(u) \subseteq U$ and for all $\op$-vertices 
$v$ in $U$ there exists a vertex $w \in \Out(v) \cap U$.
For each $\pl$-trap~$U$ player~$\op$ has a strategy 
from each vertex of~$U$ to keep the play within~$U$, 
namely choosing an edge $(v,w)$ with $w \in \Out(v) \cap U$ 
whenever the current vertex~$v$ is in $U \cap V_\op$~\cite{Zielonka98}.
For a game graph $\game$ and a $\pl$-trap $U$ we denote by $\game[U]$ 
the game graph induced by the set of vertices~$U$. Note that given that in $\game$
each vertex has at least one outgoing edge, the same property holds for $\game[U]$.

\paragraph{Attractors.}
In a game graph $\game$, a $\pl$-\emph{attractor} $A = \at{\pl}{\game}{U}$ of a set 
$U \subseteq V$ is the set of vertices from which player~$\pl$ has a strategy 
to reach $U$ against all strategies of player~$\op$.
We have that $U \subseteq A$. 
A $\pl$-attractor can be constructed inductively as follows: Let $Z_0=U$ and 
for all $\idxa \ge 0$ let $Z_{\idxa+1} = Z_\idxa \cup \cpre{\pl}(Z_\idxa)$.
Then $A = \at{\pl}{\game}{U}= \bigcup_{\idxa\ge 0} Z_\idxa$. 
In other words, $A$ is the least fixed point of 
$f(X)= U \cup \cpre{\pl}(X)$, which provides a symbolic algorithm to 
compute attractors.  
Note that the number of $\cpre{\pl}$ operations for computing the attractor is 
bounded by $\lvert A \setminus U\rvert + 1$ 
and that from each vertex of~$A$ player~$\pl$ has a memoryless strategy that
stays within $A$ to reach~$U$ against any strategy of player~$\op$~\cite{Zielonka98}.
When every vertex has at least one outgoing edge, then the complement of 
a $\pl$-attractor is a $\pl$-trap~\cite{Zielonka98}.

\begin{example}\label{example:concepts}
	In this example we describe a trap, an attractor, and a dominion of the 
	parity game in Figure~\ref{fig:example}, an illustration is provided in 
	Figure~\ref{fig:example_concepts}. For the vertex set 
	$S = \set{d, e, f, g}$ we have the property that the player-$\po$ vertices
	$\set{e,f}$ have edges only to vertices in $S$. Further, the player-$\pe$ vertices
	$\set{d, g}$ both have an edge to some other vertex in $S$. Thus the set $S$ is a
	\emph{trap} for player~$\po$. Further, when player~$\pe$ simply follows the edge 
	to another vertex in $S$ whenever a play reaches $d$ or $g$, then at least 
	every 4th step of the play the vertex~$f$ with the highest priority~$4$ is reached.
	Thus player~$\pe$ has a winning strategy from every vertex of $S$ that stays
	in $S$, hence $S$ is also an $\pe$-\emph{dominion}. Moreover, the player-$\pe$ 
	vertices $c$ and $h$ have an edge to the set~$S$, i.e., they are in the 
	$\pe$-\emph{attractor} of $S$. As another example for a player-$\pe$ attractor, 
	consider the vertex~$f$. The vertex $d$ is in the $\pe$-attractor of $f$ because 
	it belongs to player~$\pe$ and has an edge to $f$. The vertex~$e$ is in the 
	$\pe$-attractor because its only outgoing edge goes to $d$, thus player~$\po$ 
	cannot avoid reaching $f$ from $e$. Continuing this argument, it can easily be 
	seen that the $\pe$-attractor of $f$ is equal to $\set{c, d, e, f, g, h}$. Note 
	that the set $V \setminus \set{c, d, e, f, g, h} =  \set{a, b}$ is a $\pe$-trap.\qee
\end{example}

\begin{figure}[ht]
 \centering
 \begin{tikzpicture}
\matrix[column sep=12mm, row sep=8mm]{
	\node[p1,label=60:$1$, fill=white, dotted] (a) {$a$};
	& \node[p1,label=60:$1$, dotted] (c) {$c$};
	& \node[p2,label=60:$3$] (e) {$e$};
	& \node[p1,label=60:$2$] (g) {$g$};\\
	\node[p2,label=-60:$0$, fill=white, dotted] (b) {$b$};
	& \node[p1,label=-60:$0$] (d) {$d$};
	& \node[p2,label=-60:$4$] (f) {$f$};
	& \node[p1,label=-60:$1$, dotted] (h) {$h$};\\
};
\path[arrow, bend right, dotted] 
	(a) edge (b)
	(b) edge (a);
\path[arrow] 
	(c) edge[dotted] (b)
	(b) edge[dotted] (d)
	(d) edge (f)
	(f) edge (g)
	(g) edge (e)
	(e) edge (d)
	(c) edge (d)
	(h) edge[dotted] (c)
	(h) edge (g)
	;
\end{tikzpicture}
\caption{Illustration of Example~\ref{example:concepts}.
The set of solid vertices is a trap for player~$\po$ and additionally 
a player-$\pe$ dominion. The $\pe$-attractor of this $\pe$-dominion contains the 
set itself and additionally the vertices $c$ and $h$. The $\pe$-attractor 
is itself an $\pe$-dominion and the solid edges indicate a winning strategy
for player~$\pe$. This $\pe$-attractor coincides
with the $\pe$-attractor of the vertex $f$.} 
\label{fig:example_concepts}
\end{figure}
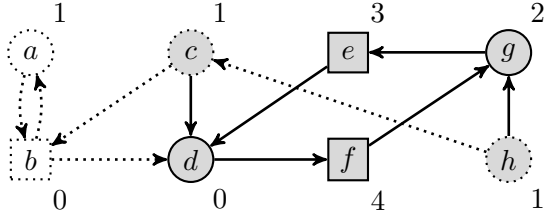

The basic algorithm for parity games uses
the following well-known properties of traps and dominions. 
Furthermore, note that in a game graph where each vertex has at least one 
outgoing edge the complement of a $\pl$-attractor is 
a $\pl$-trap~\cite[Lemma~4]{Zielonka98}
and that a player-$\pl$ dominion is also a $\op$-trap and that the 
$\pl$-attractor of a player-$\pl$ dominion is again a player-$\pl$ dominion.
Let for a specific game graph~$\game$ or a
specific parity game~$\pgame$ denote the winning set of player~$\pl$ by
$W_\pl(\game)$ and $W_\pl(\pgame)$, respectively. 
\begin{lemma}\label{lem:doms}
The following assertions hold for game graphs~$\game$ with
at least one outgoing edge per vertex 
and parity objectives. Let $\pl \in \set{\pe, \po}$ and let $U \subseteq V$.
\begin{enumerate}
\item \cite[Lemma~4.4]
{JurdzinskiPZ08} Let $U$ be a $\pl$-trap in $\game$. Then 
a $\op$-dominion in $\game[U]$ is a $\op$-dominion in $\game$.\label{sublem:winclosed}
\item \cite[Lemma~4.1]{JurdzinskiPZ08} The set $W_\pl(\game)$ is a $\pl$-dominion.\label{sublem:winsetclosed}
\item  \cite[Lemma~4.5]{JurdzinskiPZ08} Let $U$ be a subset of the winning set $W_\pl(\game)$ of player~$\pl$ and 
let $A$ be its $\pl$-attractor $\at{\pl}{\game}{U}$. Then the winning set $W_\pl(\game)$
of the player~$\pl$ is the union of $A$ and the winning set 
$W_\pl(\game[V \setminus A])$,
and the winning set $W_\op(\game)$ of the opponent~$\op$ is equal to 
$W_\op(\game[V \setminus A])$.
\label{sublem:subgraph}
\end{enumerate}
\end{lemma}
\begin{proof}
	\begin{enumerate}
		\item Player~$\pl$ cannot leave the $\pl$-trap $U$, thus 
		player~$\op$ can use the same winning strategy for the vertices of 
		the $\op$-dominion in $\game$ as in $\game[U]$.
		\item Recall that the winning sets of the two players partition the 
		vertices~\cite{Martin75}. Thus we have that as soon as the play leaves 
		the winning set of player~$\pl$,
		the opponent~$\op$ can play his winning strategy starting from the vertex 
		in his winning set that was reached. Hence the winning strategy of player~$\pl$
		for the vertices in $W_\pl(\game)$ has to ensure that only 
		vertices of $W_\pl(\game)$ are visited and thus the set $W_\pl(\game)$ is a 
		$\pl$-dominion.
		\item The set $V \setminus A$ is a $\pl$-trap~\cite[Lemma~4]{Zielonka98}. By \ref{sublem:winsetclosed} $W_\op(\game[V \setminus A])$
		is a $\op$-dominion in $\game[V \setminus A]$ and by \ref{sublem:winclosed}
		also in $\game$. Thus we have $W_\op(\game[V \setminus A]) \subseteq W_\op(\game)$.
		Since $W_\pl(\game)$ and $W_\op(\game)$ form a partition of $V$, we complete 
		the proof by showing $W_\pl(\game[V \setminus A]) \subseteq W_\pl(\game)$.
		For this we construct a winning strategy for player~$\pl$ from the vertices 
		of $W_\pl(\game[V \setminus A])$ in $\game$. As long as the play stays 
		within $V \setminus A$, player~$\pl$ follows her winning strategy 
		in $\game[V \setminus A]$. If the play reaches $A \setminus U$, she follows
		her attractor strategy to $U$. When the play reaches $U$, she follows her 
		winning strategy for $U$ in $\game$ (that exists by assumption). We have that
		player~$\pl$ wins by the winning strategy in $\game[V \setminus A]$ if 
		the play forever stays within $V\setminus A$ and by the winning strategy
		for $U$ in $\game$ if the play ever reaches a vertex of~$A$.
		\qedhere
	\end{enumerate}
\end{proof}

\section{Existing Algorithms for Parity Games}\label{sec:algoforparity}
In this section we present,
in addition to the progress measure algorithm presented in Section~\ref{sec:pmdef},
the key existing algorithms for parity games along with the main ideas for correctness. 

\subsection{Classical Algorithm}\label{sec:classic}
In the following we describe a classical algorithm for parity games 
by~\cite{Zielonka98,McNaughton93} and provide intuition for its correctness. 
Our symbolic big-step algorithm presented in Appendix~\ref{app:bigstep} uses the 
same overall structure as the classical algorithm but determines
dominions using our symbolic progress measure algorithm presented in Section~\ref{sec:pmalg}.

\begin{algorithm}
	\small
	\SetAlgoRefName{ClassicParity}
	\caption{Classical Algorithm}
	\label{alg:classic}
	\SetKwInOut{Input}{Input}
	\SetKwInOut{Output}{Output}
	\SetKw{break}{break}
	\BlankLine
	\Input{%
	  \emph{parity game} $\pgame = (\game, \prio)$, with\\ 
	  \emph{game graph} $\game = ((V, E),(\ve, \vo))$ and \\
	  \emph{priority function} $\prio: V \rightarrow [\numprio]$.
	}
	\Output
	{
	 winning sets $(\we, \wo)$ of player~$\pe$ and player~$\po$
	}
	\BlankLine
	\lIf{$\numprio = 1$}{\Return $(V, \emptyset)$}
	let $\pl$ be player~$\pe$ if $\numprio$ is odd and player~$\po$ otherwise\;
	$W_{\op} \leftarrow \emptyset$\;
	\Repeat{$W'_{\op} = \emptyset$}{
	  $\game' \gets \game \setminus \at{\pl}{\game}{P_{\numprio-1}}$\; 
	  $(\we', \wo') \leftarrow $\ref{alg:classic}$(\game', \prio)$\;
	  $A \leftarrow \at{\op}{\game}{W'_{\op}}$\;
	  $W_{\op} \leftarrow W_{\op} \cup A$\;
	  $\game \leftarrow \game \setminus A$\;
	}
	$W_{\pl} \gets V \setminus W_{\op}$\;
	\Return{$(\we, \wo)$}
\end{algorithm}

\paragraph{Informal description of classical algorithm.}
Let $\pl$ be $\pe$ if $\numprio$ is odd and $\po$ if $\numprio$ is even. 
Let $\game$ be the game graph as maintained by the algorithm.
The classical algorithm repeatedly identifies $\op$-dominions by recursive
calls for a parity game~$\pgame' = (\game', \prio)$ with one priority less that 
is obtained by temporarily removing the $\pl$-attractor of 
$P_{\numprio-1}$, i.e., the vertices with highest 
priority, from the game. In other words, the steps are as follows:
\begin{enumerate}
\item Obtain the game $\pgame'$ by removing the $\pl$-attractor of $P_{\numprio-1}$. 
\item If the winning set~$W'_\op$ 
of player~$\op$ in the parity game $\pgame'$ is non-empty, 
then its $\op$-attractor~$A$ is added to the 
winning set~$W_\op$ of $\op$ and removed from the game graph~$\game$.  
The algorithm recurses on the remaining game graph.
\item Otherwise all vertices in the parity game $\pgame'$ are winning for 
player~$\pl$. In this case the algorithm terminates and the remaining
vertices~$V \setminus W_\op$ are returned as the winning set of player~$\pl$. 
\end{enumerate}
The pseudocode of the classical algorithm is given in Algorithm~\ref{alg:classic}.

\paragraph{Key intuition for correctness.}
The correctness argument is inductive over the number of priorities 
and has the following two key aspects.
\begin{enumerate}
\item The winning set of player~$\op$ in $\pgame'$ is a $\op$-dominion in 
$(\game, \prio)$ because the vertices in $\pgame'$ form a $\pl$-trap.
Thus the attractor of the winning set of player~$\op$ in $\pgame'$ 
can be removed as part of the winning set of player~$\op$ and it suffices to 
solve the remaining game by Lemma~\ref{lem:doms}(\ref{sublem:subgraph}). 

\item If the algorithm terminates in some iteration where all vertices
in $\pgame'$ are winning for~$\pl$, then a winning strategy for 
player~$\pl$ on the remaining game can be constructed by combining her 
winning strategy in the subgame $\pgame'$ (by the inductive hypothesis over
the number of priorities as $\pgame'$ has a strictly smaller number of priorities) 
with her attractor strategy to the vertices with highest priority, 
and the fact that the set of remaining vertices $V \setminus W_\op$ is a $\op$-trap. 
\end{enumerate}
The classical algorithm can be interpreted both as an explicit 
algorithm as well as a set-based symbolic algorithm, since it only uses 
attractor computations and set operations.
The following theorem summarizes the results for the classical algorithm 
for parity games.

\begin{theorem}\label{thm:classical}
\cite{Zielonka98,McNaughton93}
Algorithm~\ref{alg:classic} computes the winning sets of parity games; 
as explicit algorithm it requires  $O(n^{\numprio-1} \cdot m)$ time and quasi-linear space; 
and as set-based symbolic algorithm it requires $O(n^\numprio)$ symbolic one step 
and set operations, and $O(\numprio)$ many sets. 
\end{theorem}  

\subsection{Sub-exponential Algorithm}\label{sec:subex}
The sub-exponential algorithm of~\cite{JurdzinskiPZ08} is based on the following
modification of the classical algorithm. 
Before the recursive call, which finds a non-empty dominion, the algorithm 
enumeratively and explicitly searches for all dominions of size at most 
$\sqrt{n}$; if it succeeds to find a dominion, then its attractor is removed
from the game; otherwise, the subsequent recursive call is 
guaranteed to find a dominion of size $> \sqrt{n}$.  
A clever analysis of the recurrence relation shows that the running time 
of the algorithm is at most $n^{O(\sqrt{n})}$, yielding the first deterministic
sub-exponential time algorithm for parity games. 
However, the algorithm is inherently explicit and enumerative (it enumerates with a brute-force
search all dominions of size at most $\sqrt{n}$).
We refer the above algorithm as \textsc{SubExp} algorithm.

\begin{theorem}\label{thm:subexp}
\cite{JurdzinskiPZ08}
Algorithm \textsc{SubExp} computes the winning sets of parity games, and 
it is an explicit algorithm that requires  $n^{O(\sqrt{n})}$ time and quasi-linear space. 
\end{theorem}  

\subsection{Big-step Algorithm}\label{sec:bigstepexpl}
The progress measure algorithm and the sub-exponential algorithm 
were combined in~\cite{Schewe17} to obtain the big-step algorithm. 
The main idea is to use the progress measure to
identify (small) dominions of size $\le \domsize+1$, for some given 
integer $\domsize \in [1, n-1]$. 
Given that an $\pe$-dominion is of size $\le \domsize+1$, player~$\pe$ must have a 
strategy from each vertex of the $\pe$-dominion to reach a vertex with an even priority by
visiting at most $\domsize$ vertices with odd priorities. 
Thus, one considers a product domain $M_\domsize \subseteq M_\game$ 
containing only the vectors of $M_\game$ whose elements sum up to at most $\domsize$.
The co-domain $M^\infty_\domsize$ of the ranking function~$\rho$ is then 
given by $M^\infty_\domsize = M_\domsize \cup \{\top\}$ and
the function $\incr(\rank)$ and $\decr(\rank)$ 
are then only defined on the restricted domain $M^\infty_\domsize$
(the $\min$ in the definitions is over $M^\infty_\domsize$ instead of $M^\infty_\game$).
Again the corresponding progress measure for a parity game 
is defined as the least simultaneous fixed point of all $\lift(.,v)$-operators. 
The identification of $\pe$-dominions from the progress measure is achieved by 
selecting those vertices whose rank is a vector, i.e., smaller than~$\top$.

\begin{lemma}\cite{Schewe17}\label{lem:schewe:progress_dominions}
For a given parity game with $n$ vertices and the progress measure $\rho$ with
co-domain $M^\infty_\domsize$ for some integer $\domsize \in [1, n-1]$,
the set of vertices $\{v \in V \mid \rho(v)<\top\}$ is an $\pe$-dominion
that contains all $\pe$-dominions with at most $\domsize+1$ vertices.
\end{lemma}

\begin{remark}\label{remark:odddom}
Note that the progress measure algorithm
determines $\pe$-dominions. We can compute $\po$-dominions
(including a winning strategy within the dominion) by adding one to each 
priority and changing the roles of the two players. If $\numprio$ has been
even before this modification, this does not increase the bound on the number 
of symbolic operations for the dominion search because the number of odd priorities,
and therefore the possible number of non-empty indices of a rank vector,
does not increase. 
\end{remark}

Combining the sub-exponential algorithm with the progress measure algorithm 
to identify small dominions gives the \textsc{BigStep} algorithm for parity games.

\begin{theorem}\cite{Schewe07}\label{thm:Schewe07}
Let $\gamma(\numprio) = \numprio/3 + 1/2 - 
4/(\numprio^2 - 1)$ for odd $\numprio$ and $\gamma(\numprio) = 
\numprio/3 + 1/2 - 1/(3 \numprio) - 4/\numprio^2$ for even $\numprio$.
Algorithm \textsc{BigStep} computes the winning sets of parity games
and it is an explicit algorithm that requires 
$O\big(m \cdot \big(\frac{\kappa \cdot n}{\numprio}\big)^{\gamma(\numprio)} \big)$  
time for some constant $\kappa$ and $O(n \cdot \numprio)$ space. 
\end{theorem}
\noindent Recently, the running time bound for the 
\textsc{BigStep} algorithm was improved further to $O\big(m \big(\frac{6 e^{5/3} n}{\numprio^2}\big)^{\gamma(\numprio)}\big)$ \cite{Schewe17}.

\section{Symbolic Big-Step Algorithm}\label{app:bigstep}

In Section~\ref{sec:pmalg} we presented a set-based symbolic algorithm to compute the 
progress measure that is also capable of determining dominions of bounded size.
Since the classical algorithm can be implemented with set-based symbolic operations,
we now show how to combine our algorithm and the classical set-based algorithm
to obtain a set-based symbolic \emph{Big-Step} algorithm for parity games (see 
Appendix~\ref{sec:bigstepexpl} for the explicit Big-Step algorithm~\cite{Schewe17}).

\begin{algorithm}
	\small
	\SetAlgoRefName{SymbolicBigStepParity}
	\caption{Set-Based Symbolic Big-Step Algorithm for Parity Games}
	\label{alg:symbigstep}
	\SetKwInOut{Input}{Input}
	\SetKwInOut{Output}{Output}
	\SetKw{break}{break}
	\BlankLine
	\Input{%
	  \emph{parity game} $P = (\game, \prio)$, with\\ 
	  \emph{game graph} $\game = ((V, E),(\ve, \vo))$,\\
	  \emph{priority function} $\prio: V \rightarrow [\numprio]$, and\\
	  \emph{parameter} $\dombound{n}{\numprio} \in [1, n] \cap \mathbb{N}$
	}
	\Output
	{
	 winning sets $(\we, \wo)$ of player~$\pe$ and player~$\po$
	}
	\BlankLine
	\lIf{$\numprio = 1$}{\Return $(V, \emptyset)$}
	let $\pl$ be player~$\pe$ if $\numprio$ is odd and player~$\po$ otherwise\;
	$W_{\op} \leftarrow \emptyset$\;
	\Repeat{$W'_{\op} = \emptyset$}{
	  \If{$\numprio > 2$}{
	$W'_{\op} \leftarrow\domalg(\game, \prio, \dombound{n}{\numprio}, \op)$\;   \label{alg:symbigstep_dominion}
	    $A \leftarrow \at{\op}{\game}{W'_{\op}}$\label{alg:symbigstep_attr1}\;
	    $W_{\op} \leftarrow W_{\op} \cup A$\;
	    $\game \leftarrow \game \setminus A$\;
	  }
	  $\game' \gets \game \setminus \at{\pl}{\game}{P_{\numprio - 1}}$\; \label{alg:symbigstep_attr2}
	  $(\we', \wo') \leftarrow $\ref{alg:symbigstep}$(\game', \prio)$\label{l:bs_nosmall}\;
	  $A \leftarrow \at{\op}{\game}{W'_{\op}}$; \label{alg:symbigstep_attr3}
	  $W_{\op} \leftarrow W_{\op} \cup A$\;
	  $\game \leftarrow \game \setminus A$\;
	}
	$W_{\pl} \gets V \setminus W_{\op}$\;
	\Return{$(\we, \wo)$}
	\BlankLine
	\BlankLine
	\SetKwProg{myproc}{Procedure}{}{}
	\myproc{$\domalg(\game, \prio, \dombound{n}{\numprio}, \op)$}{
		\If{$\op = \pe$}{
			\Return \ref{alg:SymbolicParityDominion} for $\game$, $\prio$, and $\dombound{n}{\numprio}$\;
		}\Else{
			construct the parity game $(\game', \prio')$ from $(\game, \prio)$ by
			increasing each priority by one and changing the roles of the two players\;
			\Return \ref{alg:SymbolicParityDominion} for $\game'$, $\prio'$, and $\dombound{n}{\numprio}$\;
		}
	}
\end{algorithm}

\paragraph{Iterative Winning Set Computation.}
The basic structure of Algorithm~\ref{alg:symbigstep} is the same as in 
the classical algorithm for parity games (see Appendix~\ref{sec:classic}).
Let $\pl$ be $\pe$ if $\numprio$ is odd and $\po$ if $\numprio$ is even
and assume we have $\numprio > 2$ (the cases $\numprio \le 2$, i.e., Büchi games, are
simpler). Let $\game$ be the game graph maintained by the algorithm.
The winning set of $\op$ is initialized with the empty set and then 
the algorithm searches for $\op$-dominions in a repeat-until loop.
When a $\op$-dominion is found in an iteration of the repeat-until loop, 
its $\op$-attractor is added to the winning set of $\op$ and 
removed from the game graph~$\game$.
If no $\op$-dominion is found, then the repeat-until loop terminates and 
the set of vertices in the remaining 
game graph~$\game$ is returned as the winning set of player~$\pl$.

\paragraph{Dominion Search.}
The search for $\op$-dominions is conducted in two different ways:
by Procedure~$\domalg$ and by a recursive call to a derived parity 
game with game graph $\game'$ and the priority function $\prio$ restricted
to the vertices of $\game'$. The derived parity game is denoted by 
$(\game', \prio)$ and has $\numprio - 1$ priorities. The parameter 
$\dombound{n}{\numprio}$ is used to balance the number of symbolic operations
of the two procedures.

First, all $\op$-dominions of size at most $\dombound{n}{\numprio}+1$ are found
with Procedure~$\domalg$ that uses Algorithm~\ref{alg:SymbolicParityDominion}.
Note that this algorithm determines $\pe$-dominions. We can compute $\po$-dominions
(including a winning strategy within the dominion) by adding one to each 
priority and changing the roles of the two players.

For the recursive call for the parity game~$\pgame' = (\game', \prio)$ we 
obtain the game graph~$\game'$ from $\game$ by removing the $\pl$-attractor 
of the vertices with the highest priority. Note that the vertices of $\game'$
form a $\pl$-trap in $\game$. The winning set of player~$\op$ in $\pgame'$
is a $\op$-dominion in the original parity game. 

\paragraph{Number of Iterations.}
The following lemma shows that the number of iterations of the repeat-until 
loop is bounded by $O(n / \dombound{n}{\numprio})$. 
This holds because in each but the last iteration
the union of the dominion identified by $\domalg$ 
and the dominion identified by the recursive call
is larger than $\dombound{n}{\numprio}+1$;
otherwise, the union of the two dominions, which is itself a dominion,
would already have been identified solely by Procedure~$\domalg$ and the 
algorithm would have terminated.

\begin{lemma}[\cite{Schewe17}]\label{lem:iterations}
	Let $\dombound{n}{\numprio}$ be the parameter of Algorithm~\ref{alg:symbigstep}.
	In each but the last iteration of the repeat-until loop
	at least $\dombound{n}{\numprio} + 2$ vertices are removed, and thus
	there are at most $\lfloor n / (\dombound{n}{\numprio} + 2) \rfloor + 1$ 
	iterations.
\end{lemma}

\begin{proof}
	Consider a fixed iteration of the repeat-until loop.
	Let $W''_\op$ be the union of the $\op$-dominion identified by Procedure~$\domalg$
	and the $\op$-dominion identified by the recursive call. If 
	$\lvert W''_\op \rvert \le \dombound{n}{\numprio} + 1$, then by 
	Lemma~\ref{lem:small_dominions_space_efficient} the $\op$-dominion $W''_\op$ is 
	identified by Procedure~$\domalg$ and therefore the recursive call  
	returned the empty set and the algorithm terminates.
	Otherwise we have $\lvert W''_\op \rvert > \dombound{n}{\numprio} + 1$,
	which can happen in at most $\lfloor n / (\dombound{n}{\numprio} + 2) \rfloor$ 
	iterations of the repeat-until loop since the vertices of $W''_\op$ are 
	removed from the maintained game graph before the next iteration.
\end{proof}

\paragraph{Analysis of number of symbolic operations.}
We first analyze the number of symbolic operations
of Algorithm~\ref{alg:symbigstep} for instances with 
an arbitrary numbers of priorities~$\numprio$
and then give better bounds for the case $\numprio \leq \sqrt{n}$.
For the former we follow the analysis of~\cite{JurdzinskiPZ08} for a similar explicit 
algorithm. There we choose the parameter $\dombound{n}{\numprio}$ to depend
(only) on the number of vertices of the game graph for which Procedure~$\domalg$ is called.
In the analysis for $\numprio \leq \sqrt{n}$
the parameter $\dombound{n}{\numprio}$ depends 
on the number of vertices in the \emph{input game graph} and the number 
of priorities in the game graph for which the procedure is called. 

\begin{lemma}\label{lem:bigstep_runtimes_bigc}      
	Let $\dombound{n'}{\numprio}=\lceil \sqrt{2n'}\rceil-2$ 
	where $n'$ is the number of vertices in $\game$ when Procedure~$\domalg$
	is called in Algorithm~\ref{alg:symbigstep}. With this choice for 
	$\dombound{n'}{\numprio}$ and parity games with $n$ vertices, 
	Algorithm~\ref{alg:symbigstep} requires $n^{O(\sqrt{n})}$ symbolic one-step 
	and set operations and $O(n)$ many sets.
\end{lemma}
\begin{proof}
  Let $T(n)$ denote the number of the symbolic one-step operations of Algorithm~\ref{alg:symbigstep} 
  when called for a game graph with $n$~vertices.
  We instantiate $\dombound{n}{\numprio}$ as $\dombound{n}{\numprio}=\lceil 
  \sqrt{2n}\rceil-2$. 
  First, for a game graph of size~$n$,
  by Lemma~\ref{lem:small_dominions_space_efficient}, we can compute a dominion in
  $O\left( \numprio \cdot n \cdot \binom {\dombound{n}{\numprio} + \lfloor \numprio/2 \rfloor} {\dombound{n}{\numprio}}
    \right)$. 
    This can be bounded with $O(n^{\dombound{n}{\numprio}+2})$
    as follows (using $\numprio \le n$).
    First we apply Stirling's approximation 
    $(\dombound{n}{\numprio}/e)^{\dombound{n}{\numprio}} \le \dombound{n}{\numprio}!$ and we 
    distinguish whether  $\dombound{n}{\numprio} \geq \lfloor \numprio/2 \rfloor$ 
    or $\dombound{n}{\numprio} \leq \lfloor \numprio/2 \rfloor$. We have
    \begin{align*}
     \binom {\dombound{n}{\numprio} + \lfloor \numprio/2 \rfloor} {\dombound{n}{\numprio}} 
     \leq \frac{(\dombound{n}{\numprio} + \lfloor \numprio/2 \rfloor)^{\dombound{n}{\numprio}}}{\dombound{n}{\numprio}!}
     \leq 
    \left( \frac{(\dombound{n}{\numprio} + \lfloor \numprio/2 \rfloor)e} {\dombound{n}{\numprio}} \right)^{\dombound{n}{\numprio}} \,,
    \end{align*}
    and (a) if $\dombound{n}{\numprio} \geq \lfloor \numprio/2 \rfloor$, then
    \begin{align*}
    \left( \frac{(\dombound{n}{\numprio} + \lfloor \numprio/2 \rfloor)e} {\dombound{n}{\numprio}} \right)^{\dombound{n}{\numprio}} 
    \leq
    \left( \frac{2\dombound{n}{\numprio} \cdot e} {\dombound{n}{\numprio}} \right)^{\dombound{n}{\numprio}} 
    \leq
    \left( 2e \right)^{\dombound{n}{\numprio}}\,,
    \end{align*}
    (b) if $\dombound{n}{\numprio} \leq \lfloor \numprio/2 \rfloor$, then
    $$
    \left( \frac{(\dombound{n}{\numprio} + \lfloor \numprio/2 \rfloor)e} {\dombound{n}{\numprio}} \right)^{\dombound{n}{\numprio}} \leq
      \left( \frac{2 \lfloor \numprio/2 \rfloor \cdot e} {\dombound{n}{\numprio}} \right)^{\dombound{n}{\numprio}}\,.
    $$
    In both cases we obtain $O(n^{\dombound{n}{\numprio}+2})$.\par
  Let $\Gamma = \dombound{n}{\numprio} +2 = \lceil \sqrt{2n}\rceil$. 
  To solve an instance of size $n$, the algorithm
  (1) calls the subroutine for computing a dominion, which is in $O(n^\Gamma)$; 
  (2) then makes a recursive call to Algorithm~\ref{alg:symbigstep}, with an instance of 
  size at most size $n-1$
  (it removes the largest parity and thus at least one vertex); and
  (3) finally processes the instance in which
  all winning vertices from the previous two steps have been removed,
  i.e., an instance of size at most $n-\Gamma$ (Lemma~\ref{lem:iterations}).
  Thus, we obtain the following recurrence relation for $T(n)$
  \[
  T(n) \leq O(n^\Gamma) + T(n-1)+T(n-\Gamma)
  \]
  for $\Gamma= \lceil \sqrt{2n}\rceil$.
  This coincides with the recurrence relation for the explicit algorithm
  and thus we get $T(n)=n^{O(\sqrt{n})}$~\cite[Theorem 8.1]{JurdzinskiPZ08}.
  Finally, as in Algorithm~\ref{alg:symbigstep} the number of symbolic set operations is 
  at most a factor of $n$ higher than the number of symbolic one-step operations,
  we obtain the claim by $n \cdot T(n) = n^{O(\sqrt{n})}$.
  
  Now let us consider the space requirements of Algorithm~\ref{alg:symbigstep}.
  Computing a dominion, by Lemma~\ref{lem:small_dominions_space_efficient}, 
  can be done with $O(n)$ many sets.
  The recursion depth is bounded by $\numprio$ and thus also by $n$, and 
  for each instance on the stack we just need a constant number of sets 
  to store the current subgraph and the winning sets of the two players.
  Finally, also attractor computations can be done with a constant number of sets.
  That is, Algorithm~\ref{alg:symbigstep} uses $O(n)$ many sets,
  independent of the value of $\dombound{n}{\numprio}$.
\end{proof}

To provide an improved bound for $\numprio \le \sqrt{n}$ we follow the 
analysis of~\cite{Schewe07}; see~\cite{Schewe17} for a recent improvement of 
this analysis. We use the following definitions from~\cite{Schewe07}:
let $\gamma(\numprio) = \numprio/3 + 1/2 - 
4/(\numprio^2 - 1)$ for odd $\numprio$ and $\gamma(\numprio) = 
\numprio/3 + 1/2 - 1/(3 \numprio) - 4/\numprio^2$ for even $\numprio$, and let
$\beta(\numprio) = \gamma(\numprio)/(\lfloor\numprio / 2\rfloor + 1)$.
With these definitions we have 
$\beta(\numprio) \in [1/2, 7/10] \text{ for all } \numprio \ge 3$, 
$\gamma(\numprio) \in [\numprio/3, \numprio/3+1/2]$ and
\begin{align}
 \gamma(\numprio) = \gamma(\numprio - 1) + 1 - \beta(\numprio - 1)\,, \label{eqn:gamma}\\
 \beta(\numprio-1) \cdot \lceil \numprio/2 \rceil = \gamma(\numprio - 1)\,. \label{eqn:beta}
\end{align}
We first present a shorter proof for the simplified bound of 
$O(n^{1+\gamma(\numprio)})$ symbolic one-step operations and then a more extensive
analysis to show that, for some constant~$\kappa$, $O(n \cdot (\kappa \cdot n /
\numprio)^{\gamma(\numprio)})$ symbolic one-step operations are sufficient.
We bound the number of symbolic set operations with an additional factor 
of $O(n)$. In both cases the parameter $\dombound{n}{\numprio}$ depends on 
number of vertices in the input game graph and the number of 
priorities~$\numprio$ in the parity game that is maintained by the algorithm.

\begin{lemma}[Simplified Analysis of Number of Symbolic Operations]\label{lem:bigstep_runtimesimple}
	For parity games with $n$ vertices and $\numprio$~priorities
	Algorithm~\ref{alg:symbigstep} takes $O(n^{1+\gamma(\numprio)})$ 	
	symbolic one-step operations and $O(n^{2+\gamma(\numprio)})$ symbolic set 
	operations for $2 <\numprio \leq \sqrt{n}$ and $O(n^2)$ symbolic one-step and
	set operations for $\numprio = 2$. Moreover, it needs $O(n)$ many sets.
\end{lemma}
\begin{proof}
	By the proof of Lemma~\ref{lem:bigstep_runtimes_bigc} the algorithm only 
	stores $O(n)$ many sets, independent of the value of $\dombound{n}{\numprio}$, 
	and thus it only remains to show the upper bound on the number of symbolic operations.
	We focus on the bound for the number of symbolic one-step operations,
	the number of symbolic set operations is higher only for the subroutine
	\ref{alg:SymbolicParityDominion}, where the difference is at most a factor
	of $O(n)$, by Lemma~\ref{lem:small_dominions_space_efficient}.
	For $\numprio = 2$ note that in each iteration of the repeat-until loop at
	least one vertex is removed from the game, thus there can be at most $O(n)$
	iterations. In each iteration there are attractor computations,
	a recursive call for parity games with one priority, and set operations,
	which all can be done with $O(n)$ symbolic operations\footnote{Note that 
	while the attractor computations in line~\ref{alg:symbigstep_attr3} only need $O(n)$ 
	symbolic operations over the whole algorithm, 
	the attractor computations in line~\ref{alg:symbigstep_attr2} might require $O(n)$ 
	symbolic operations in each iteration.},
	thus at most $O(n^2)$ symbolic operations are needed for $\numprio = 2$.
	
	By Lemma~\ref{lem:iterations} the repeat-until loop can have
	at most $\left \lfloor \frac{n}{(\dombound{n}{\numprio} + 2)} \right \rfloor + 1$ iterations.
	We show the claimed number of symbolic one-step operations for $\numprio \ge 3$
	by induction over~$\numprio$.
	For the base case of $\numprio = 3$ let $\dombound{n}{3} = n$. In this
	case there is only one iteration of the repeat-until loop. The 
	call to the progress measure procedure and the 
	recursive call for parity games with one priority less, i.e., Büchi games,
	both take $O(n^2)$ symbolic one-step operations. Thus the total 
	number of symbolic one-step operations for $\numprio = 3$ is bounded 
	by $O(n^{1 + \gamma(3)}) = O(n^2)$. For $\numprio=3$ the number of
	sets $S_\rank$ is only $O(n)$, therefore we do not have to use the sets $C^i_x$ to 
	achieve linear space and thus the number of symbolic set operations 
	is also $O(n^2)$.
	
	For the induction step we differentiate between the number of vertices~$n_0$
	in the first (i.e., non-recursive) call to Algorithm~\ref{alg:symbigstep}
	and the number of vertices~$n'$ in the parity game that is maintained by 
	the algorithm.
	We set the parameter $\dombound{n}{\numprio}$ according to $n_0$ (and to at most
	$n'$) to maintain the property $\numprio \le \sqrt{n_0}$
	in all recursive calls.
	
	Let $\dombound{n}{\numprio} = \min\big(\lceil n_0^{\beta(\numprio-1)} \rceil,
	n'\big)$ for $\numprio > 3$. 
	By Lemma~\ref{lem:iterations} the number of iterations in 
	Algorithm~\ref{alg:symbigstep} is bounded by $O(n_0^{1-\beta(\numprio - 1)})$.
	By Lemma~\ref{lem:small_dominions_space_efficient} each call to \ref{alg:SymbolicParityDominion} needs only 
	\begin{equation*}
	O\left(
	    \numprio \cdot n_0 \cdot 
	    \binom { \lceil n_0^{\beta(\numprio-1)} \rceil + \lfloor \numprio/2 \rfloor} { \lceil n_0^{\beta(\numprio-1)} \rceil} 	
	\right)
	\end{equation*} many symbolic one-step operations.
	The binomial coefficient can be bounded using $\beta(\numprio-1) \geq 1/2$,
	Stirling's approximation of $(\numprio/e)^\numprio \le \numprio!$, and 
	$3 < \numprio \leq \sqrt{n_0}$.
	\begin{align*}
	  \binom { \lceil n_0^{\beta(\numprio-1)} \rceil + \lfloor \numprio/2 \rfloor} { \lceil n_0^{\beta(\numprio-1)} \rceil} &\leq
	  \frac{(\lceil n_0^{\beta(\numprio-1)} \rceil + \lfloor \numprio/2 \rfloor)^{\lfloor \numprio/2 \rfloor}}{ \lfloor \numprio/2 \rfloor!}
	  \le \left(\frac{2 e \cdot (\lceil n_0^{\beta(\numprio-1)} \rceil + \lfloor \numprio/2 \rfloor)}{\numprio - 1}\right)^{\lfloor \numprio/2 \rfloor}\,,\\
	  &\leq\left(\frac{2 e \cdot n_0^{\beta(\numprio-1)} + e \cdot \numprio}{\numprio - 1}\right)^{\lfloor\numprio / 2\rfloor} 
	  \leq \left(\frac{4 e \cdot n_0^{\beta(\numprio-1)}}{\numprio} \right)^{\lfloor\numprio / 2\rfloor}
	\end{align*}
	thus, the number of symbolic one-step operations in a call to  \ref{alg:SymbolicParityDominion} is bounded by
	$O(n_0 \cdot n_0^{\beta(\numprio - 1) \lfloor\numprio / 2\rfloor})$.
	Hence, 
	we can bound the total number of symbolic one-step operations for all calls to  \ref{alg:SymbolicParityDominion}
	with $O(n_0 \cdot  n_0^{\beta(\numprio - 1) \lfloor\numprio / 2\rfloor}
	\cdot n_0^{1-\beta(\numprio - 1)})$. 
	Using
	\eqref{eqn:beta},
	we have $\beta(\numprio - 1)\lfloor\numprio / 2\rfloor
	\le \beta(\numprio - 1)\lceil\numprio / 2\rceil = \gamma(\numprio - 1)$
	and can bound the total number of symbolic one-step operations
	of the calls to the progress measure
	procedure with $O(n_0 \cdot n_0^{\gamma(\numprio - 1) + 1 - \beta(\numprio - 1)})$,
	which by \eqref{eqn:gamma} is equal to $O(n_0^{1+\gamma(\numprio)})$.
	By the induction assumption we further have that the total number of symbolic
	one-step operations for all recursive calls is 
	$O(n_0^{1-\beta(\numprio - 1) + 1 + \gamma(\numprio - 1)})
	= O(n_0^{1 + \gamma(\numprio)})$, which concludes the proof.
\end{proof}
\begin{lemma}[Number of Symbolic Operations]\label{lem:bigstep_runtime}
	For parity games with $n$ vertices and $2 < \numprio \le \sqrt{n}$ priorities
	Algorithm~\ref{alg:symbigstep} takes, for some constant~$\kappa$,
	$O\big(n \cdot (\frac{\kappa \cdot n}{\numprio})^{\gamma(\numprio)}\big)$ symbolic one-step 
	operations and
	$O\big(n^2 \cdot (\frac{\kappa \cdot n}{\numprio})^{\gamma(\numprio)}\big)$ symbolic
	set operations.
	Moreover, it needs $O(n)$ many sets.
\end{lemma}
\begin{proof} 
By the proof of Lemma~\ref{lem:bigstep_runtimes_bigc} the algorithm only stores $O(n)$ many sets independent of the value of $\dombound{n}{\numprio}$, 
and thus it only remains to show the upper bound on the number of symbolic operations.
We show the bound for the number of symbolic one-step operations,
	the number of symbolic set operations is higher only for the subroutine
	\ref{alg:SymbolicParityDominion}, where the difference is at most a factor
	of $O(n)$ by Lemma~\ref{lem:small_dominions_space_efficient}.

  We differentiate between the number of vertices~$n_0$
	in the first (i.e.\ non-recursive) call to Algorithm~\ref{alg:symbigstep}
	and the number of vertices~$n'$ in the parity game that is maintained by 
	the algorithm.
	We will set the parameter $\dombound{n}{\numprio}$ according to $n_0$ (and 
	to at most $n'$) to maintain the property $\numprio \le \sqrt{n_0}$
	in all recursive calls.
	
Similar to~\cite{Schewe07}, we use $\dombound{n_0}{\numprio} = \min\left(
\lceil 2 \sqrt[3]{\numprio} \cdot n_0^{\beta(\numprio-1)} \rceil, n'\right)$ to balance 
the number of symbolic operations to compute dominions with the number of symbolic 
operations in the recursive calls.
We will frequently apply \emph{Stirling's approximation} by which we have
\begin{itemize}
\item[ (S1)]\ $\big(\frac{\numprio}{e}\big)^\numprio \le \numprio!$, and
\item[ (S2)]\ $\numprio! \in O\big(\big(\frac{\numprio}{\hat{\kappa}}\big)^\numprio\big)$ for all $\hat{\kappa} < e$.
\end{itemize}
We first show that proving the following claim is sufficient.

\begin{claim}
The number of symbolic one-step operations
is bounded by $\kappa_1 \cdot n_0 \cdot \frac{(\kappa_2 
\cdot n_0)^{\gamma(\numprio)}} { \sqrt[3]{\numprio !}}$ for some constants~$\kappa_1$
and $\kappa_2$. 
\end{claim}
The claim implies the lemma since we have for $\numprio \ge 3$
\begin{align*}
	\kappa_1 \cdot n_0 \cdot \frac{(\kappa_2 
\cdot n_0)^{\gamma(\numprio)}}{\sqrt[3]{\numprio !}}
&\le 
\kappa_1 \cdot n_0 \cdot (\kappa_2 
\cdot n_0)^{\gamma(\numprio)} \cdot \left(\frac{e}{\numprio}\right)^{\frac{\numprio}{3}} 
\le 
\kappa_1 \cdot n_0 \cdot (\kappa_2 
\cdot n_0)^{\gamma(\numprio)} \cdot \left(\frac{e}{\numprio}\right)^{\gamma(\numprio)-\frac{1}{2}}\,,\\
&\le 
\kappa_1 \cdot n_0 \cdot \left(\frac{\kappa_2 \cdot e \cdot \kappa_3
\cdot n_0}{\numprio}\right)^{\gamma(\numprio)}
\end{align*}
for some constant $\kappa_3 > 1$. 
The first inequality is by (S1) and the second by the definition of $\gamma(\numprio)$.
Thus it remains to prove the above claim.
\smallskip

We prove the claim by induction over $\numprio$, where the base case is 
provided by Lemma~\ref{lem:bigstep_runtimesimple} for, e.g., $\numprio = 3$.
Assume the claim holds for parity games with $\numprio - 1$ priorities; we will
show that the claim then also holds for parity games with $\numprio$ priorities.
For $\numprio > 2$ the number of symbolic operations 
per iteration is dominated by (a)
the recursive call for parity games with one priority less
and (b) the call to Procedure~$\domalg$. To bound the total number 
of symbolic operations for (a) and (b), we use that by Lemma~\ref{lem:iterations}
the number of iterations in Algorithm~\ref{alg:symbigstep} is 
bounded by 
\begin{equation*}
\left \lfloor \frac{n_0}{\dombound{n_0}{\numprio} + 2} \right \rfloor + 1 \le 
\frac{n_0^{1-\beta(\numprio - 1)} }{ 2\sqrt[3]{\numprio}} + 1 \leq
2 \cdot \frac{n_0^{1-\beta(\numprio - 1)}}{3\sqrt[3]{\numprio}}\,.
\end{equation*}

(a) We first bound the total number of symbolic one-step operations for the 
recursive call. By the induction hypothesis the number of symbolic one-step operations
in a single recursive call is
bounded by 
\begin{equation*}
\kappa_1 \cdot n_0 \cdot \frac{(\kappa_2 \cdot n_0)^{\gamma(\numprio - 1)}
} {\sqrt[3]{(\numprio - 1)!}}\,. 
\end{equation*}
By the definition of $\gamma(\numprio)$ (see also the proof of 
Lemma~\ref{lem:bigstep_runtimesimple}) we can bound the total number of symbolic
one-step operations for all recursive calls with
\begin{equation}\label{eq:timerec}
\begin{split}
	\kappa_1 \cdot n_0 \cdot \frac{(\kappa_2 \cdot n_0)^{\gamma(\numprio - 1)}}{
\sqrt[3]{(\numprio - 1)!}} \cdot \frac{2 \cdot n_0^{1-\beta(\numprio - 1)}}{3\sqrt[3]{\numprio}} 
&\le \kappa_1 \cdot n_0 \cdot \frac{2 \cdot (\kappa_2 \cdot n_0)^{\gamma(\numprio)}}{3\sqrt[3]{\numprio!}}\,.
\end{split}
\end{equation}

(b) Thus it remains to bound the total number of symbolic one-step operations in the 
calls to Procedure~$\domalg$.
By Lemma~\ref{lem:small_dominions_space_efficient} the number of symbolic one-step operations in 
a call to Procedure~$\domalg$ is bounded by 
$O\left( \numprio \cdot n_0 \cdot \binom { \dombound{n_0}{\numprio} + \lfloor \numprio/2 \rfloor} {\lfloor \numprio/2 \rfloor} \right)$.
	The binomial coefficient can be bounded by using $\beta(\numprio-1) \geq 1/2$, 
	Stirling's approximation (S1), and $3 < \numprio \leq \sqrt{n_0}$.
	We have
	\begin{align*} 
	  \binom { \lceil 2 \sqrt[3]{\numprio} \cdot n_0^{\beta(\numprio-1)} \rceil + \lfloor \numprio/2 \rfloor} { \lfloor \numprio/2 \rfloor} 
	  & \leq \frac{(\lceil 2 \sqrt[3]{\numprio} \cdot n_0^{\beta(\numprio-1)} \rceil 
	  + \lfloor \numprio/2 \rfloor)^{\lfloor \numprio/2 \rfloor}}{ \lfloor \numprio/2 \rfloor!}\,,\\
	  &\leq \left(\frac{e \cdot (\lceil 2 \sqrt[3]{\numprio} \cdot n_0^{\beta(\numprio-1)} \rceil+ \lfloor \numprio/2 \rfloor)}{\lfloor \numprio/2 \rfloor}\right)^{\lfloor \numprio/2 \rfloor}\,,\\
	  &\leq 
	  \left(\frac{ (1 + 1/4) \cdot e \cdot \lceil 2 \sqrt[3]{\numprio} \cdot n_0^{\beta(\numprio-1)} \rceil}{ \lceil (\numprio-1)/2 \rceil} \right)^{\lceil \frac{\numprio-1}{2} \rceil}\,.
	\end{align*}

With this bound on the binomial coefficient we have that the
number of symbolic one-step operations in one call to Procedure~$\domalg$ is bounded by 
\begin{align*}
\kappa_4 \cdot \numprio \cdot n_0 \cdot 
\left(
  \frac{
    7 \lceil 
      2 \sqrt[3]{\numprio} n_0^{\beta(\numprio-1)} 
    \rceil
  }
  {
      \numprio-1
  }
\right)^{
\lceil \frac{\numprio-1}{2} \rceil}
\end{align*}
for some constant $\kappa_4$. To obtain a bound on 
the total number of symbolic one-step operations for all calls to Procedure~$\domalg$,
we multiply that with the number of iterations 
$2 \cdot n_0^{1-\beta(\numprio - 1)} / (3\sqrt[3]{\numprio})$
and split the further analysis into two parts.

\noindent (1) First, we consider the factor 
$	\lceil n_0^{\beta(\numprio-1)} \rceil^{\lceil \frac{\numprio-1}{2} \rceil} 
	\cdot n_0^{1-\beta(\numprio-1)}$
depending on $n_0$.
We upper bound the factor by using $c \leq \sqrt{n_0}$ and $\beta(\numprio-1) \geq 1/2$, 
and then apply \eqref{eqn:gamma}.
\begin{align*}
	\lceil n_0^{\beta(\numprio-1)} \rceil^{\lceil \frac{\numprio-1}{2} \rceil} 
	\cdot n_0^{1-\beta(\numprio-1)}
	&\le \sqrt{e} \cdot n_0^{\beta(\numprio-1) \cdot \lceil \frac{\numprio}{2} \rceil} 
	\cdot n_0^{1-\beta(\numprio-1)}
	= \sqrt{e} \cdot n_0^{\gamma(\numprio)}\,.
\end{align*}
(2)~Second, we consider the factor 
$\numprio \cdot \left(7 {\lceil 2 \sqrt[3]{\numprio}\rceil}/
{( \numprio-1 )}\right)^{\lceil \frac{\numprio-1}{2} \rceil}$
depending on $\numprio$ and show that,
when assuming\footnote{Notice that for $\numprio \leq 63$ the factor is bounded by a constant anyway.} $\numprio \geq 64$,
we have
\begin{align*}
\numprio \cdot \left(\frac{7 \lceil 2 \sqrt[3]{\numprio}\rceil}
{\numprio-1}\right)^{\lceil \frac{\numprio-1}{2} \rceil}
&\le \numprio \cdot \left( \frac{\kappa_5 \sqrt[3]{\numprio}}
{\numprio - 1} \right)^{\frac{\numprio-1}{2}} 
\le \numprio \cdot
\left( \frac{\kappa_6 \sqrt[3]{\numprio-1}}{\numprio-1} \right)^{\frac{\numprio-1}{2}}\,,\\
&\le \numprio \cdot
\left( \frac{\sqrt{\kappa_6}}{\sqrt[3]{\numprio-1}} \right)^{\numprio-1}
\le \numprio \cdot
\left( \frac{4}{\sqrt[3]{\numprio-1}} \right)^{\numprio-1}\,.
\end{align*}
with 
$\kappa_5= 63/4$, and
$\kappa_6= \kappa_5 \cdot \sqrt[3]{64/63} \approx 15.83$.
Notice that in the above equation 
we exploited the assumption $\numprio \geq 64$ several times. 
First, to bound $7 \lceil 2 \sqrt[3]{\numprio}\rceil$ by $\kappa_5 \sqrt[3]{\numprio}$,
second, to obtain $\left( \frac{\kappa_5 \sqrt[3]{\numprio}} {\numprio - 1} \right) \leq 1$,
and, finally, to bound $\kappa_5 \sqrt[3]{\numprio}$ by $\kappa_6 \sqrt[3]{\numprio-1}$.

Now let $\kappa_7$ be such that $\kappa_7 / (\numprio - 1)! \ge (\hat{\kappa} / (\numprio - 1))^{\numprio - 1}$ for some $\hat{\kappa}$ slightly smaller than $e$, such that we can apply (S2).
Notice that we also can bound $4^3/\hat{\kappa}$ by $24$ for $\hat{\kappa}$ close to $e$. 
We have
\begin{align*}
	\numprio \cdot
\left( \frac{4 \cdot \sqrt[3]{\hat{\kappa}}}{\sqrt[3]{\hat{\kappa}} \cdot \sqrt[3]{(\numprio-1}} \right)^{\numprio-1}
&\le \numprio \cdot
\frac{\sqrt[3]{\kappa_7} \cdot 24^{(\numprio-1)/3}}{\sqrt[3]{(\numprio-1)!}}\,
\le \frac{\kappa_8 \cdot \kappa_9^{\gamma(\numprio)}}{\sqrt[3]{(\numprio-1)!}}\,,
\end{align*}
for appropriately chosen constants $\kappa_8$ and $\kappa_9$.
Putting parts~(1), ~(2) and the not yet considered factor of 
$\kappa_4 \cdot 2/(3\sqrt[3]{\numprio})$ together, 
we have that the total number of symbolic one-step operations for all calls to Procedure~$\domalg$ is bounded by
\begin{equation*}
\begin{split}
	\kappa_4 \cdot \numprio \cdot n_0 \cdot 
\left(\frac{7 \lceil 2 \sqrt[3]{\numprio} n_0^{\beta(\numprio-1)} \rceil
}{\lceil \frac{\numprio-1}{2} \rceil}\right)^{
\lceil \frac{\numprio-1}{2} \rceil} \cdot 
\frac{2 \cdot n_0^{1-\beta(\numprio - 1)}}{3 \sqrt[3]{\numprio}}
&\le \kappa_1 \cdot n_0 \cdot
\frac{(\kappa_2 \cdot n_0)^{\gamma(\numprio)}}{3 \sqrt[3]{\numprio!}}
\end{split}
\end{equation*}
for $\kappa_1 \ge  2\sqrt{e} \cdot \kappa_4 \cdot \kappa_8$ and $\kappa_2 \ge \kappa_9$.
Together with Equation~\eqref{eq:timerec} this concludes the proof.
\end{proof}

\paragraph{Correctness.} 
The correctness of Algorithm~\ref{alg:symbigstep} is as follows: 
(i)~the algorithm correctly identifies dominions for player $\op$ 
(Lemma~\ref{lem:small_dominions_space_efficient}), and then removes them
(Lemma~\ref{lem:doms}(\ref{sublem:subgraph})), 
and it suffices to solve the remaining game graph; 
and (ii)~when no vertices are removed, then no dominion is found, and the computation is exactly
as the classical algorithm, and from the correctness of the classical algorithm it follows
that the remaining vertices are winning for player $\pl$.
We present a correctness proof below to make the paper 
self-contained.

\begin{lemma}[Correctness]\label{lem:bigstep_correctness}
	The algorithm correctly computes the winning sets of player~$\pe$ and player~$\po$
	for parity games with $\numprio \ge 1$ priorities.
\end{lemma}
\begin{proof}
	The proof is by induction over~$\numprio$. For $\numprio = 1$
	all vertices are winning for player~$\pe$ and the algorithm correctly returns
	$\we = V$ and $\wo = \emptyset$. Assume the algorithm correctly computes
	the winning sets for parity games with $\numprio - 1$ priorities. We show that
	this implies the correctness for $\numprio$ priorities. 
	Let $\pl$ be $\pe$ if $\numprio$ is odd and $\po$ otherwise. 
	We first show (1) that each vertex of $W_\op$ is indeed winning for player~$\op$
	and then (2) that each vertex of $W_\pl$ is winning for player~$\pl$; 
	this is sufficient to prove the lemma by $W_\op \cup W_\pl = V$.
	
	For (1) recall that the
	algorithm repeatedly computes $\op$-dominions~$D$ and their $\op$-attractor~$A$,
	adds $A$ to the winning set of player~$\op$ and recurses on the parity game
	with $A$ removed. By Lemma~\ref{lem:doms}(\ref{sublem:subgraph}) 
	we have that this approach is correct if the sets $D$ are indeed $\op$-dominions.
	By the soundness of \ref{alg:SymbolicParityDominion} by 
	Lemma~\ref{lem:small_dominions_space_efficient} we have that $D$ is
	a $\op$-dominion when computed with \ref{alg:SymbolicParityDominion};
	it remains to show the soundness of determining a dominion $W_\op'$ 
	of player~$\op$ by the recursive call to
	\ref{alg:symbigstep} on the parity game $(\game', \prio)$ with 
	one priority less (line~\ref{l:bs_nosmall}). This follows by
	Lemma~\ref{lem:doms}(\ref{sublem:winclosed}) from $\game'$ 
	not containing a vertex with priority~$\numprio-1$ and the complement 
	of a $\pe$-attractor being a $\pe$-trap~\cite[Lemma~4]{Zielonka98}.
	
	To show completeness, i.e., $W_\op(P) \subseteq W_\op$, we describe 
	a winning strategy for player~$\pl$ on the vertices of $W_\pl$. Since the 
	set $W_\pl$ is the set of remaining vertices after the
	removal of $\op$-attractors, the set $W_\pl$ is a $\op$-trap.
	Let $\game^*$ be the game graph as in the last iteration of the algorithm, i.e., $\game^* = \game[W_\pl]$. Furthermore, 
	let $Z$ be the set of vertices in $\game^*$ with priority $\numprio-1$,
	and let $\game' = \game^* \setminus  \at{\pl}{\game^*}{Z}$.
	Since the algorithm has terminated, we have that player~$\pl$ wins on all 
	vertices of $\game'$ in the parity game~$(\game', \prio)$. The winning strategy
	for player~$\pl$ for the vertices of~$W_\pl$ is as follows: for vertices
	of $Z$ the player choses an edge to some vertex in $W_\pl$;
	for vertices of $\at{\pl}{\game^*}{Z} \setminus Z$
	the player follows her attractor strategy to $Z$; and for the vertices of $\game'$
	the player plays according to her winning strategy in~$(\game', \prio)$.
	Then in a play starting from $W_\pl$ either $Z$ is visited
	infinitely often or the play remains within $\game'$ from 
	some point on; in both cases player~$\pl$ wins with the given strategy.
\end{proof}

\section{Strategy Construction}\label{app:strategy_construction}
Here we discuss how our algorithms can be extended to also compute winning strategies
within the same bounds for the number of symbolic operations and sets required.

\paragraph{Obtaining Strategies from the Symbolic Progress Measure.}
Let us first consider Algorithm~\ref{alg:SymbolicParityDominion}.
Given the set representation of the progress measure returned by Algorithm~\ref{alg:SymbolicParityDominion}
the winning strategy for player~$\pe$ can be computed
with $O(n)$ symbolic one-step operations and $O(c \cdot n^2)$ symbolic set operations.

To this end, notice that the values of the progress measure immediately give a winning strategy of 
player~$\pe$. That is, for a vertex $v$ in the winning set of $\pe$ a winning strategy of $\pe$
picks an arbitrary successor $w$ of $v$, with $\rho(w) \leq_{\prio(v)} \rho(v)$  if $v$ is of even priority or
with $\rho(w) <_{\prio(v)} \rho(v)$  if $v$ is of odd priority.

In the symbolic setting we can compute the strategy as follows. 
We maintain a set $S$ of the vertices in $V_\pe  \setminus V_\top$ that already have a strategy,
and first initialize $S$ as the empty set.
We then iterate over all vertices $v \in V \setminus V_\top$, i.e., over all vertices 
that are winning for player~$\pe$, and do the following.
\begin{itemize}
  \item First, compute the exact rank $\rho(v)$ of the vertex $v$ by
  checking for which of the sets $C^i_x$ the intersection with $\set{v}$
  is not empty with $O(n)$ set operations.\footnote{Alternatively we could use binary search on the 
  sets $S_\rank$ at the cost of increasing the set operations for this step
  by a factor of $O(c \cdot \log n)$.}  
  \item Second, for each $0 \leq \ell < \numprio$ compute the set 
	$S_{\incr_\ell(\rho(v))}$ with $O(n)$  many symbolic set operations per set.
  \item Then compute $\cpre{\pe}(\{v\})$ and consider the 
	 sets $\cpre{\pe}(\{v\}) \cap S_{\incr_\ell(\rho(v))} \cap V_\pe \cap P_\ell$, 
for $0 \leq \ell < \numprio$.
	From each vertex contained in one of the sets a winning strategy for player~$\pe$
	can move to $v$ and thus, for each of these vertices not already in yhe set $S$,
	we fix $v$ as the strategy of player~$\pe$ and add it to $S$.
\end{itemize}
In each iteration of the above algorithm we use only one $\cpre{\pl}$ operation and 
$O(c \cdot n)$ symbolic set operations and as we have $O(n)$ iterations,
the claim follows.

\paragraph{Computing Strategies with the Big-Step Algorithm.}
Now consider the symbolic big step algorithm \ref{alg:symbigstep} that computes the winning sets for a parity game.
One can also compute the actual strategies for both players by using 
\ref{alg:symbigstep} and within the bounds of Theorem~\ref{thm:bigstep}.
This is by the following observations.
\begin{itemize}
 \item Whenever a vertex is added to a winning set because of \ref{alg:SymbolicParityDominion} (line~\ref{alg:symbigstep_dominion}) then 
       we can obtain a winning strategy for that vertex
       within the same algorithmic bounds (as outlined above).
 \item For $\numprio=1$ player $\pe$ can just pick any successor.
       It is then  easy to construct a strategy with $O(|V|)$ many symbolic operations.
 \item For a vertex that is added to a winning set because of a recursive call to \ref{alg:symbigstep} (line~\ref{l:bs_nosmall}),
       we can use the strategy that was computed in the recursive call.
 \item When a vertex is added to a winning set because of the computation 
			of the $\pl$-attractor~$A$ in lines~\ref{alg:symbigstep_attr1} 
			and~\ref{alg:symbigstep_attr3} we can compute the winning strategy with $O(\lvert A \rvert)$ symbolic operations as follows.
       In the computation of the attractor we keep track of the vertices $A_{\idxa-1}$ added in the previous iteration of the attractor
       computation and whenever a vertices of $V_\op$ is added to the attractor we identify a successor for each
       of them. 
       To this end, for each vertex in $v \in A_{\idxa-1}$ we compute $\pre(\{v\}) \cap A_\idxa \cap V_\op$ and 
       fix $v$ as players $\op$ choice for all these vertices.
       As each vertex of $A$ is in exactly one set $A_\idxa$, we only need $O(\lvert A \rvert)$ many operations,
       and as the attractor computation itself needs $O(\lvert A \rvert)$ many $\cpre{\pl}$ operations this does not increase 
       the bound for the number of symbolic operations.
\end{itemize}
\end{document}